\documentclass[reprint,footinbib,prx]{revtex4-1}

\usepackage{amsfonts}
\usepackage{amssymb}
\usepackage{amsmath}
\usepackage{amsthm}
\usepackage{mathrsfs}
\usepackage{bm}
\usepackage{graphicx}
\usepackage{enumerate}
\usepackage{multirow}
\usepackage{hyperref}
\usepackage{float}
\usepackage{framed}
\usepackage{adjustbox}
\usepackage{appendix}
\usepackage{lipsum}

\hypersetup{colorlinks=true}

\newtheoremstyle{definitionstyle}
  {0.2cm}
  {0.2cm}
  {\it}
  {}
  {\it\bfseries}
  {.}
  { }
  {\thmname{#1}\thmnumber{~#2}\thmnote{~(#3)}}
  
\newtheoremstyle{nameddefinitionstyle}
  {\baselineskip\@plus.2\baselineskip\@minus.2\baselineskip}
  {\baselineskip\@plus.2\baselineskip\@minus.2\baselineskip}
  {}
  {}
  {\bfseries}
  {.}
  { }
  {\thmnote{#3}}
  
\newtheoremstyle{framednameddefinitionstyle}
  {0.2cm}
  {0.2cm}
  {\it}
  {}
  {\it\bfseries}
  {.}
  { }
  {\thmnote{#3}}
  
\newtheoremstyle{theoremstyle}
  {0.2cm}
  {0.2cm}
  {\it}
  {}
  {\it\bfseries}
  {.}
  { }
  {\thmname{#1}\thmnumber{~#2}\thmnote{~(#3)}}
  
\newtheoremstyle{framedtheoremstyle}
  {\baselineskip\@plus.2\baselineskip\@minus.2\baselineskip}
  {\baselineskip\@plus.2\baselineskip\@minus.2\baselineskip}
  {\sl}
  {}
  {\bfseries}
  {.}
  { }
  {\thmname{#1}\thmnumber{~#2}\thmnote{~(#3)}}
  
\newtheoremstyle{proofstyle}
  {\baselineskip\@plus.2\baselineskip\@minus.2\baselineskip}
  {\baselineskip\@plus.2\baselineskip\@minus.2\baselineskip}
  {}
  {}
  {}
  {.}
  { }
  {\textsc{\thmname{#1}\thmnote{~#3}}}

\theoremstyle{theoremstyle}

\newtheorem{prp}{Proposition}

\newtheorem{cor}{Corollary}

\theoremstyle{framedtheoremstyle}

\theoremstyle{definitionstyle}

\theoremstyle{definitionstyle}

\theoremstyle{definitionstyle}

\theoremstyle{definitionstyle}

\theoremstyle{nameddefinitionstyle}

\theoremstyle{framednameddefinitionstyle}
\newtheorem*{framednameddef}{}

\theoremstyle{proofstyle}

\theoremstyle{definitionstyle}

\newcommand{\fromto}{\rightarrow}

\newcommand{\xfromto}[1]{\xrightarrow{#1}}

\newcommand{\ZZZ}{\mathbb{Z}}
\newcommand{\NNN}{\mathbb{N}}

\newcommand{\RRR}{\mathbb{R}}
\newcommand{\CCC}{\mathbb{C}}

\renewcommand{\SS}{\mathbf{S}}

\newcommand{\T}{\mathcal{T}}

\newcommand{\W}{\mathcal{W}}

\DeclareMathOperator{\Map}{Map}
\DeclareMathOperator{\identity}{id}

\DeclareMathOperator{\image}{im}

\DeclareMathOperator{\Aut}{Aut}

\DeclareMathOperator{\Tor}{Tor}
\DeclareMathOperator{\Ext}{Ext}

\DeclareMathOperator{\kernel}{ker}

\newcommand{\coloneq}{\mathrel{\mathop:}=}

\newcommand{\homotopic}{\simeq}

\newcommand{\isomorphic}{\cong}

\newcommand{\paren}[1]{\left( #1 \right)}
\newcommand{\angles}[1]{\left\langle #1 \right\rangle}
\newcommand{\brackets}[1]{\left[ #1 \right]}
\newcommand{\braces}[1]{\left\{ #1 \right\}}

\newcommand{\ket}[1]{\left|#1\right\rangle}

\def\ind\hspace{0.2in}

\newcommand{\SPT}{\operatorname{\mathcal{SPT}}}

\newcommand{\wSPT}{\operatorname{w\mathcal{SPT}}}

			\newcommand{\e}[1]{\begin{align}{#1}\end{align}}	
		
		\newcommand{\la}[1]{\label{#1}}

		\newcommand{\q}[1]{Eq.\ (\ref{#1})}
		
		\newcommand{\s}[1]{Sec.\ \ref{#1}}
		\newcommand{\fig}[1]{Fig.\ \ref{#1}}

\newcommand{\bk}{\boldsymbol{k}}

\newcommand{\bE}{\boldsymbol{E}}

\newcommand{\bB}{\boldsymbol{B}}

\newcommand{\cals}{{\cal S}}

\newcommand{\bpm}{\begin{pmatrix}}
\newcommand{\epm}{\end{pmatrix}}

\newcommand{\bal}{\begin{align}}

\newcommand{\sma}[1]{\scriptscriptstyle{#1}}

\newcommand{\noc}{n_{\sma{{o}}}}

\begin{document}

\title{Organizing symmetry-protected topological phases by layering and symmetry reduction: a minimalist perspective}
\author{Charles Zhaoxi Xiong$^1$}
\email{zxiong@g.harvard.edu}
\author{A. Alexandradinata$^2$}
\affiliation{$^1$Department of Physics, Harvard University, Cambridge, Massachusetts 02138, USA}
\affiliation{$^2$Department of Physics, Yale University, New Haven, Connecticut 06520, USA}
\date{\today}

\begin{abstract}
It is demonstrated that fermionic/bosonic symmetry-protected topological (SPT) phases across different dimensions and symmetry classes can be organized using geometric constructions that increase dimensions and symmetry-reduction maps that change symmetry groups. Specifically, it is shown that the interacting classifications of SPT phases with and without glide symmetry fit into a short exact sequence, so that the classification with glide is constrained to be a direct sum of cyclic groups of order 2 or 4. Applied to fermionic SPT phases in the Wigner-Dyson class AII, this implies that the complete interacting classification in the presence of glide is ${\mathbb Z}_4{\oplus}{\mathbb Z}_2{\oplus}{\mathbb Z}_2$ in 3 dimensions. In particular, the hourglass-fermion phase recently realized in the band insulator KHgSb must be robust to interactions. Generalizations to spatiotemporal glide symmetries are discussed.
\end{abstract}

\maketitle

\section{Introduction\label{sec:introduction}}

The recent intercourse between band theory, crystalline symmetries, and topology has been highly fruitful in both theoretical and experimental laboratories. The recent experimental discovery \cite{Ma_discoverhourglass} of hourglass-fermion surface states in the material class KHgSb \cite{Hourglass,Cohomological} heralds a new class of topological insulators (TIs) protected by glide symmetry \cite{ChaoxingNonsymm,unpinned,Shiozaki2015,Nonsymm_Shiozaki,Poyao_mobiuskondo,Ezawa_hourglass, shiozaki_review,singlediraccone} 
-- a reflection composed with a translation by half the lattice period. Despite the manifold successes of band theory, electrons are fundamentally interacting. To what extent are topological phases predicted from band theory robust to interactions?

In this work, we demonstrate that the question of (a) robustness to interactions is intimately linked to two other seemingly unrelated questions: (b) how glide-symmetric topological phases can be constructed by layering lower-dimensional topological phases, and (c) in what ways can the classification of topological phases be altered by the inclusion of glide symmetry.

In fact, question (c) is very close in spirit to the types of questions asked in a symmetry-based classification of solids: how many ways are there to combine discrete translational symmetry with rotations and/or reflections to form a space group -- the full symmetries of a crystalline solid? This has been recognized as a group extension problem, and its solution through group cohomology has led to the classification of 230 space groups of 3D solids \cite{hiller1986crystallography}. Here, we are proposing that the same mathematical structure ties together (a-c). More precisely, we are proposing a short exact sequence of abelian groups, which classify gapped, \emph{interacting} phases of matter, also known as symmetry-protected topological (SPT) phases \cite{SPT_origin, Wen_Definition}, that carries the information of (a-c).

\begin{figure}[t]
\centering
\includegraphics[width=3in]{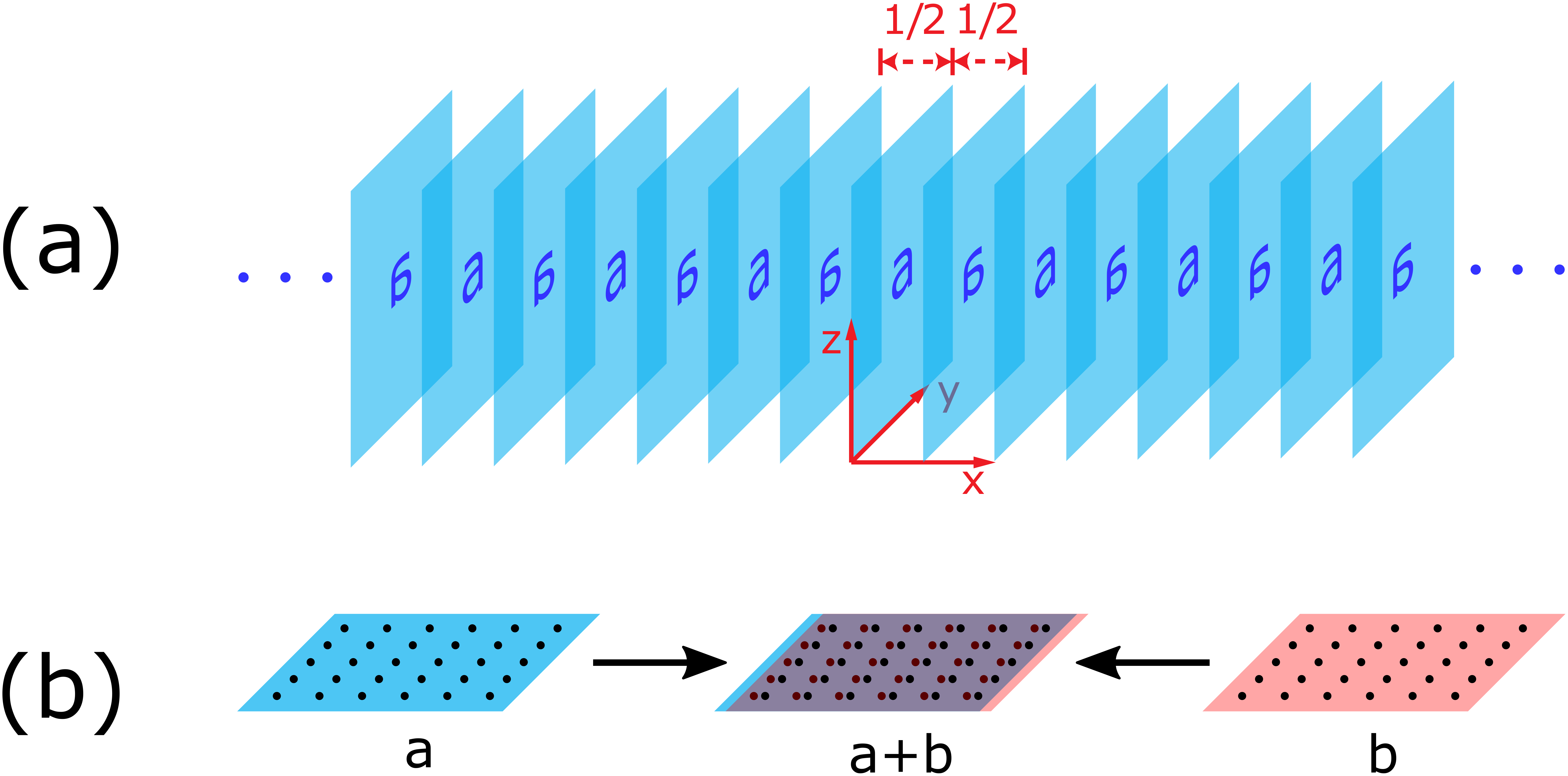}
\caption{(a) The ``alternating-layer construction" repeats a given $G$-symmetric system $a$ and its mirror image in an alternating fashion to produce a one higher-dimensional system that respects glide symmetry in addition to $G$. (b) The ``stacking" operation combines \emph{two} systems $a$ and $b$ respecting a given symmetry into a new system of the \emph{same} dimension respecting the \emph{same} symmetry. Illustration is given for particular dimensions but the constructions are general.
}
\label{fig:alternating_layer_construction}
\label{fig:stacking}
\end{figure}

In the symmetry class of hourglass fermions, i.e., spin-orbit-coupled solids with charge-conservation [$U(1)$], time-reversal ($\T$, with $\T^2 = -1$ on single fermions), and glide symmetries \footnote{More precisely, the hourglass-fermion phase belongs to a nonsymmorphic space group which includes at least one glide symmetry.}, there is a pair of consecutive maps between abelian groups classifying TIs in two and three dimensions,
\begin{equation}
\ZZZ_2 \xfromto{\times 2} \ZZZ_4 \xfromto{\rm mod~ 2} \ZZZ_2,\label{SES_AII}
\end{equation}
that can be viewed as a non-interacting analog of our short exact sequence.  The nontrivial element of the first $\ZZZ_2$, which distinguishes the two phases of 2D TIs that respect $\T$ and $U(1)$,  can be realized by a 2D quantum spin Hall (QSH) system \cite{kane2005B, kane2005A}.
By placing \emph{decoupled} copies of a QSH system on all planes of constant $x\in \ZZZ$, and its mirror image on all $x\in \ZZZ + 1/2$ planes, one constructs a 3D system that respects the glide symmetry $(x,y,z)\mapsto(x+1/2,-y,z)$.\phantom{\setcounter{footnote}{10}\footnote{Since there is no known generalization of the $\ZZZ_4$ invariant to disordered systems, we assume discrete translational symmetry in all three directions at first and quotient out phases that can be obtained by layering lower-dimensional phases in the $y$- and $z$-directions in the end \cite{Nonsymm_Shiozaki}. In the same spirit, both $\ZZZ_2$'s in Eq.\,(\ref{SES_AII}) are obtained after quotienting out layered phases; that is, they are the strong \cite{kitaev2009} classifications.}} This ``alternating-layer construction" (see Fig.\,\ref{fig:alternating_layer_construction}) takes one from the first $\ZZZ_2$ to $\ZZZ_4$ -- in particular the QSH phase to the hourglass fermion phase  \cite{Nonsymm_Shiozaki,Ezawa_hourglass} -- where $\ZZZ_4$ distinguishes the four phases of 3D TIs that respect glide in addition to $\T$ and $U(1)$ \cite{Nonsymm_Shiozaki,Note11}. By dropping the glide symmetry constraint, one can in turn view a 3D TI respecting glide, $\T$, and $U(1)$ as an element of the second $\ZZZ_2$, which is the strong classification of 3D TIs respecting $\T$ and $U(1)$ but not necessarily glide \cite{moore2007,fu2007b,Rahul_3DTI}. In this ``symmetry-forgetting" process, certain distinct classes in the $\ZZZ_4 = \braces{0,1,2,3}$ classification are identified: classes 0 and 2 (resp.\,1 and 3) can be connected to each other if glide symmetry is not enforced. This gives the second, ``glide-forgetting" map between $\ZZZ_4$ and $\ZZZ_2$ in Eq.\,(\ref{SES_AII}). We stress that glide forgetting only conceptually expands the space of allowed Hamiltonians by letting go of the glide constraint, and does not involve an immediate, actual perturbation to a particular system that is under consideration.

\setcounter{footnote}{20}

Our short exact sequence of abelian groups classifying SPT phases works in essentially the same manner as Eq.\,(\ref{SES_AII}) but with the non-interacting classification replaced by the classification of bosonic or fermionic SPT phases. 
In its full generality, the sequence applies to all symmetries $G$, including those that are represented \cite{Jiang_sgSPT, Thorngren_sgSPT} antiunitarily, and all dimensions $d$, where the analog of glide is
\begin{equation}
(x_1, x_2, x_3, \ldots, x_d) \mapsto (x_1 + 1/2, -x_2, x_3, \ldots, x_d).\label{glide_definition}
\end{equation}
Writing $\ZZZ$ for the symmetry generated by Eq.\,(\ref{glide_definition}), the existence of the short exact sequence implies that all $d$-dimensional $\ZZZ\times G$-protected SPT phases must have order 1, 2, or 4 and that their classification must be a direct sum of $\ZZZ_4$'s {and/or} $\ZZZ_2$'s, where the order of an SPT phase is defined with respect to a ``stacking" operation (imagine interlaying two systems \emph{without} coupling them; see Fig.\,\ref{fig:stacking}) 
\footnote{In the literature on non-interacting fermionic phases, the stacking operation is also known as the {``Whitney sum,"} which refers to the direct sum of vector spaces corresponding to two sets of fermion-filled bands over the same Brillouin torus. Note that a direct sum of single-particle Hilbert spaces corresponds to a tensor product of Fock spaces or many-body Hilbert spaces.}
`$+$' that makes the set of $d$-dimensional $G$-protected SPT phases into an abelian group. The short exact sequence also implies that, in general, not all $d$-dimensional $G$-protected SPT phases have glide-symmetric representatives and that a necessary and sufficient condition for such representatives to exist is for the given $G$-protected SPT phase to square to the trivial phase, where the square of an SPT phase $[a]$ is by definition $2[a] \coloneq [a]+[a]$. Note that this implication is non-obvious because certain $G$-protected SPT phases are known to be incompatible with certain symmetries outside the group $G$: e.g., a Chern insulator that conserves charge is not compatible with time reversal. From the perspective gained through our short exact sequence, it is then not surprising, in the symmetry class of the hourglass-fermion phase, that there exist four non-interacting 3D phases in the presence of glide and that the nontrivial 3D $\ZZZ_2$ TI, which squares to the trivial phase, can be made glide-symmetric \cite{Cohomological}.

In fact, by combining our general result with the proposed complete classifications of 2D \cite{2dChiralBosonicSPT_erratum} and 3D \cite{WangChong_3DSPTAII} fermionic SPT phases in the Wigner-Dyson class AII, we can show that the complete classification of 3D fermionic SPT phases with an additional glide symmetry must be $\ZZZ_4 \oplus \ZZZ_2 \oplus \ZZZ_2$, such that the first summand can be identified with the $\ZZZ_4$ in Eq.\,(\ref{SES_AII}).
We will do so in two steps. First, we will argue that the hourglass-fermion phase is robust to interactions using a corollary of the general result and known arguments \cite{qi_spincharge,essin2009} for the robustness of QSH systems and 3D TIs without glide. Assuringly, the same conclusion was recently drawn in Ref.\,\cite{Lu_sgSPT} through the construction of an anomalous surface topological order. Then, we will show that the exactness of the sequence
\begin{equation}
0 \fromto \ZZZ_2 \fromto \mbox{?} \fromto \ZZZ_2 \oplus \ZZZ_2 \oplus \ZZZ_2 \fromto 0,
\end{equation}
where $0$ denotes the trivial group (also written $\ZZZ_1$),
is simply constraining enough that $\ZZZ_4\oplus\ZZZ_2\oplus\ZZZ_2$ is the \emph{unique} solution that is compatible with the robustness of the hourglass-fermion phase.

We will derive our general result within a bare-bones, minimalist framework that one of us developed \cite{Xiong} based on Kitaev's argument that the classification of SPT phases should carry the structure of a generalized cohomology theory \cite{Kitaev_Stony_Brook_2011_SRE_1, Kitaev_Stony_Brook_2013_SRE, Kitaev_IPAM}.
The framework assumes minimally that SPT phases form abelian groups satisfying certain axioms, and applies to all existing non-dimension-specific proposals for the classification of SPT phases \cite{Wen_Boson, Wen_Fermion, Kapustin_Boson, Freed_SRE_iTQFT,  Freed_ReflectionPositivity, Kitaev_Stony_Brook_2011_SRE_1, Kitaev_Stony_Brook_2013_SRE, Kitaev_IPAM}.
The axioms provide for the switching from one symmetry group to another and from one dimension to another, which makes the derivation of our short exact sequence possible.
The results mentioned above are far from an exhaustive list of implications of the short exact sequence, which we will elaborate upon in this work.

This paper is organized as follows. In Sec.\,\ref{sec:minimalist_framework}, we will review the minimalist framework. In Sec.\,\ref{sec:hourglass_fermions}, we will argue that the hourglass-fermion phase and its square roots are robust to interactions. In Sec.\,\ref{sec:general_relations}, we will deal with SPT phases with glide more systematically. We will give a more precise definition of SPT phases, derive our general result, and explore its implications. In Sec.\,\ref{sec:physical_picture}, we will break down the general result into individual statements and offer the physical intuition behind some of them. In Sec.\,\ref{sec:applications}, we will apply the general result to 3D fermionic SPT phases in Wigner-Dyson classes A and AII, where the complete classification with glide will be derived. In Sec.\,\ref{sec:computations}, we will do the same for bosonic SPT phases for a variety of symmetries. In Sec.\,\ref{sec:discussions}, we will discuss generalized, spatiotemporal glide symmetries, the difference between glide and pure translation, a 0-dimensional proof that time reversal gives the inverse of an SPT phase, and the consistency among the arguments for the robustness of various phases. We conclude in \s{sec:summaryoutlook} with a summary of our results and a tentative discussion of the potential existence of relations for other spatial symmetries (e.g.\,reflection) that may be derived from the Hypothesis.

\section{The minimalist framework\label{sec:minimalist_framework}}

Existing proposals for the classification of SPT phases fall roughly into two categories: those of a constructive nature \cite{Wen_Boson, Cirac, Wen_Fermion,Kitaev_Stony_Brook_2011_SRE_1, Kitaev_Stony_Brook_2013_SRE, Kitaev_IPAM, Jiang_sgSPT, Wang_intrinsic_fermionic, Huang_dimensional_reduction}, and those that postulate topological invariants \cite{Kapustin_Boson, Kapustin_Fermion, Freed_SRE_iTQFT,  Freed_ReflectionPositivity, Else_edge, 2dFermionGExtension, Thorngren_sgSPT, Wang_Levin_invariants, Wang_intrinsic_fermionic, SO_infty}. Let us understand this through two examples. The group cohomology proposal \cite{Wen_Boson} for bosonic SPT phases is of a constructive nature, in that each group cocycle serves as the input to the construction of a concrete lattice model. It was shown that two equivalent group cocycles give rise to lattice models belonging to the same bosonic SPT phase, so one can say that given any group cohomology class -- that is an equivalence class of group cocycles -- a concrete bosonic SPT phase can be constructed. In principle, this correspondence can be either incomplete (the construction does not produce all phases) or degenerate (distinct cohomology classes correspond to the same phase). The former is manifested by the existence of the 3D $E_8$-phase \cite{3dBTScVishwanathSenthil, 3dBTScWangSenthil, 3dBTScBurnell}, while the latter was suggested in Ref.\,\cite{Kapustin_Boson} to occur in 6 (spatial) dimensions.

The Freed-Hopkins proposal \cite{Freed_SRE_iTQFT,Freed_ReflectionPositivity} is one that postulates topological invariants in that, given any physical system, say a lattice model, by suitably taking the long-distance limit, one expects to obtain a topological field theory of some kind, which can in turn be classified. A deformation of the physical system should not affect the resulting topological field theory, which can hence be thought of as a topological invariant. In principle, this topological invariant can be either incomplete (cannot resolve all distinct phases) or superfluous (not all values are physically realizable) \cite{Freed_SRE_iTQFT,Freed_ReflectionPositivity}. While long-distance limits are rarely explicitly taken, topological responses (e.g. Hall conductance) and derived topological invariants (e.g. ground-state degeneracy) can often be computed from the topological field theories themselves.

\begin{figure}[t]
\centering
\includegraphics[width=2.3in]{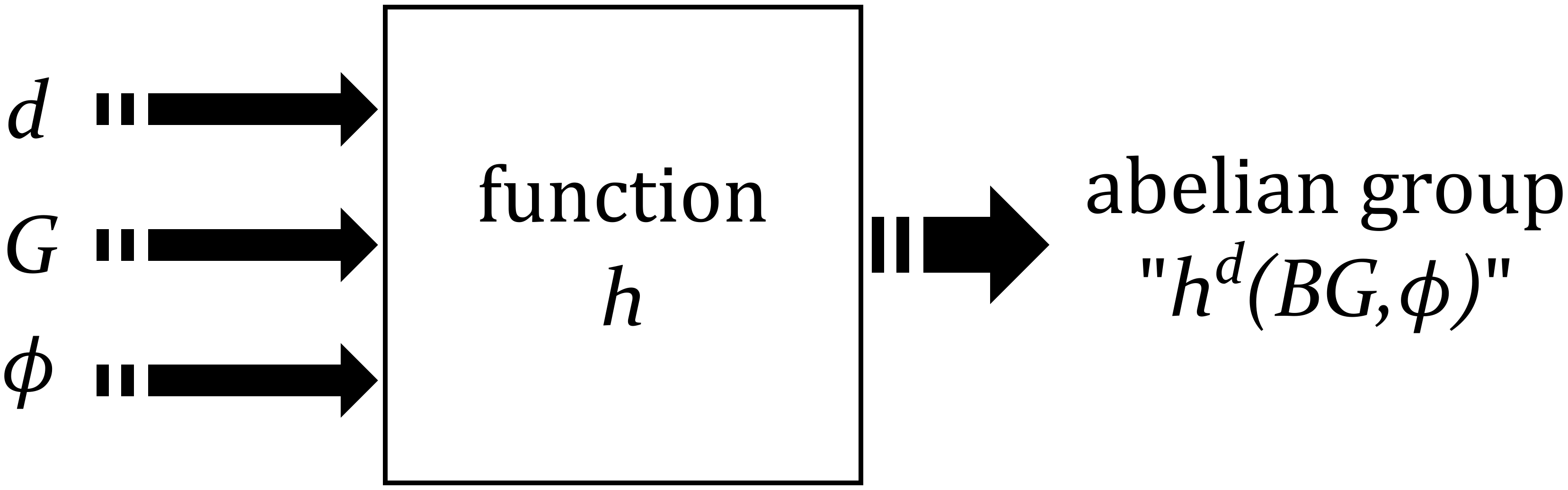}
\caption{The minimalist framework. By relinquishing all other ingredients, one places the focus entirely on the formal classification: a function $h$ that takes a triple $(d,G,\phi)$ as the input and returns an abelian group as the output.}
\label{fig:minimalist_framework}
\end{figure}

In contrast, the minimalist framework adopted in Ref.\,\cite{Xiong} is stripped of all such constructive procedures and invariant- or response-computing protocols. It places the focus entirely on the formal classification: a function $h$ that takes a triple $(d,G,\phi)$ as the input and returns a set, which we shall denote by $h^d(BG,\phi)$, as the output (see Fig.\,\ref{fig:minimalist_framework}). Of course the ultimate goal is to find the correct $h$, the one whose output gives precisely the classification of SPT phases when $d$ is interpreted as the spatial dimension and $\paren{G,\phi}$ as the symmetry. (Here, $G$ is the symmetry group, and $\phi: G \fromto \braces{\pm 1}$ is a map that keeps track of whether a symmetry reverses the orientation of spacetime \cite{Jiang_sgSPT, Thorngren_sgSPT} and is often dropped from the notation when no confusion may arise.)
Two approaches to this challenging task seem natural: either we look for the best humanly possible approximations to the correct $h$, or we make as few assumptions about $h$ and keep things as general as possible. Existing proposals \cite{Wen_Boson, Cirac, Wen_Fermion,Kitaev_Stony_Brook_2011_SRE_1, Kitaev_Stony_Brook_2013_SRE, Kitaev_IPAM, Jiang_sgSPT,Kapustin_Boson, Kapustin_Fermion, Freed_SRE_iTQFT,  Freed_ReflectionPositivity, Else_edge, 2dFermionGExtension, Thorngren_sgSPT, Huang_dimensional_reduction, Wang_Levin_invariants, Wang_intrinsic_fermionic, SO_infty} represent the former approach. Here and in Ref.\,\cite{Xiong}, we advocate for the latter.

The minimalist assumption we make about the function $h$ is that it satisfies the Eilenberg-Steenorod axioms for generalized cohomology with local coefficients \cite{Hatcher, DavisKirk, Adams1, Adams2}. We assume there is one such $h$ for bosonic SPT phases and one for fermionic SPT phases (the subtlety of fermion parity will be addressed in App.\,\ref{app:TGCH}), but we will omit such qualifiers as ``bosonic" and ``fermionic" since the discussion applies to both. The axioms endow each $h^d(BG,\phi)$ with the structure of an abelian group, and we demand that this matches the abelian group structure of SPT phases defined by stacking (see Fig.\,\ref{fig:stacking}). The axioms also imply that homomorphism $f: (G_1,\phi_1) \fromto (G_2,\phi_2)$ will naturally give rise to maps $h^d(BG_2,\phi_2) \fromto h^d(BG_1, \phi_1)$, and we demand that these match the corresponding symmetry-forgetting maps in case $f$ is an inclusion of subgroup. This minimalist assumption, which we shall refer to as the Twisted Generalized Cohomology Hypothesis, is based on Kitaev's argument that the classification of SPT phases should carry the structure of a generalized cohomology theory \cite{Kitaev_Stony_Brook_2011_SRE_1, Kitaev_Stony_Brook_2013_SRE, Kitaev_IPAM}. In particular, it was argued that $h^d(BG,\phi)$ can be written as the set of homotopy classes of $G$-equivariant maps from $BG$ to the space of $d$-dimensional short-range entangled states \cite{Kitaev_Stony_Brook_2011_SRE_1, Kitaev_Stony_Brook_2013_SRE, Kitaev_IPAM}.
The fact that all existing non-dimension-specific proposals \cite{Wen_Boson, Wen_Fermion, Kapustin_Boson, Kapustin_Fermion, Freed_SRE_iTQFT,  Freed_ReflectionPositivity, Kitaev_Stony_Brook_2011_SRE_1, Kitaev_Stony_Brook_2013_SRE, Kitaev_IPAM} for the classification of SPT phases satisfy the Hypothesis further supports its validity \cite{Xiong}.

The strategy that was adopted in Ref.\,\cite{Xiong} and will now be pursued is this: we will try to derive as many results as possible while assuming \emph{only} the Twisted Generalized Cohomology Hypothesis. Fortunately, the axioms for generalized cohomology theories are substantial enough for this approach to be useful. 
While we may not get any classification directly this way, we can nevertheless reveal relations between classifications in different dimensions for different symmetries that may be otherwise nontrivial; by combining these relations with known results for certain symmetries in certain dimensions, we can then derive the classification for other symmetries in other dimensions that we are interested in. Thanks to the minimalism in our premise, the relations will not depend on any details that are specific to particular proposals for the classification of SPT phases.

A review of generalized cohomology theories and a more precise formulation of the Twisted Generalized Cohomology Hypothesis can be found in App.\,\ref{app:twisted_generalized_cohomology}.

\section{Robustness of hourglass fermions to interactions\label{sec:hourglass_fermions}}

Let us apply generalized cohomology to 3D spin-orbit-coupled, time reversal-invariant TIs and their interacting analogues. Due to spin-orbit coupling, time reversal necessarily carries half integer-spin representation, i.e.\,it squares to $-1$ on single fermions. With the addition of glide symmetry, the two well-known classes of 3D TIs, which are distinguished by a $\ZZZ_2$ index $\nu_0\in \braces{0,1}$ \cite{fu2007b,moore2007,Rahul_3DTI,Inversion_Fu}, subdivide into four classes distinguished by a $\ZZZ_4$ invariant $\chi \in \braces{0,1,2,3}$ \cite{Nonsymm_Shiozaki, AA_Z4}. Of the four classes, $\chi=2$ corresponds to the hourglass fermion phase. We will call the other two nontrivial phases the ``square roots" of the hourglass-fermion phase, since the $\ZZZ_4$ invariant adds under stacking and $1+1 \equiv 3 + 3 \equiv 2 \mod 4$. In this section we will argue
\begin{enumerate}[(i)]
\item that the square roots of the hourglass-fermion phase are robust to interactions, and \label{list_1_2}

\item that the hourglass-fermion phase is robust to interactions. \label{list_1_3}
\end{enumerate}
Note that (\ref{list_1_3}) implies (\ref{list_1_2}), for if the $\chi= 1$ or 3 phase was unstable to interactions, then so would two decoupled copies of itself, which represent the $\chi = 2$ phase. Nevertheless, we will dedicate a separate subsection to (\ref{list_1_2}), which can be justified by an independent magnetoelectric response argument, as a consistency check.

\subsection{Robustness of the square roots of the hourglass-fermion phase\label{subsec:robustness_square_roots}}

In this subsection, we argue that the square roots ($\chi = 1, 3$) of the hourglass-fermion phase are robust to interactions. More precisely, we argue that they cannot be connected to the trivial phase ($\chi=0$) by turning on interactions that preserve the many-body gap and the glide, $U(1)$, and $\T$ symmetries.

We begin by noting that the $\ZZZ_4$ and $\ZZZ_2$ classifications with and without glide symmetry are related as
\begin{equation}
\chi \equiv \nu_0 \mod 2,\la{glideforget}
\end{equation}
which is supported by the following heuristic argument \footnote{This argument was first presented in Ref.\ \cite{Nonsymm_Shiozaki}}. Recall that $\nu_0$ counts the parity of the number of surface Dirac fermions \cite{fu2007b,moore2007,Rahul_3DTI,Inversion_Fu}. The glide symmetry assigns to each Dirac fermion a chirality according to its glide representation [compare Figs.\,\ref{fig:z4surfacestates}(b) and (c)]. Unless symmetry is broken, we cannot \cite{Hourglass} fully gap out surface states that carry two positively-chiral Dirac fermions [\fig{fig:z4surfacestates}(d)]. However, two positively-chiral Dirac fermions can be deformed into two negatively-chiral fermions [\fig{fig:z4surfacestates}(d)$\rightarrow$(e)$\rightarrow$(f)$\rightarrow$(g)]. We thus expect four topologically distinct classes, which we distinguish by a $\ZZZ_4$ invariant $\chi$ that counts the number of chiral Dirac fermions mod 4. Since both $\chi$ and $\nu_0$ count the number of surface Dirac fermions, Eq.\,(\ref{glideforget}) follows. This argument was made with representatives of $\chi=\pm 1,\pm 2$ whose surface states are Dirac fermions situated at the Brillouin zone center (point $2$ in Fig.\,\ref{fig:z4surfacestates}); more generally, the surface states form a nontrivial connected graph over the bent line $0123$ \cite{Nonsymm_Shiozaki,Hourglass}. This motivates a more general proof of \q{glideforget}, which we present in App.\,\ref{app:relationship_Z4_Z2}.

\begin{figure}[tb]
\centering
\includegraphics[width=8.6 cm]{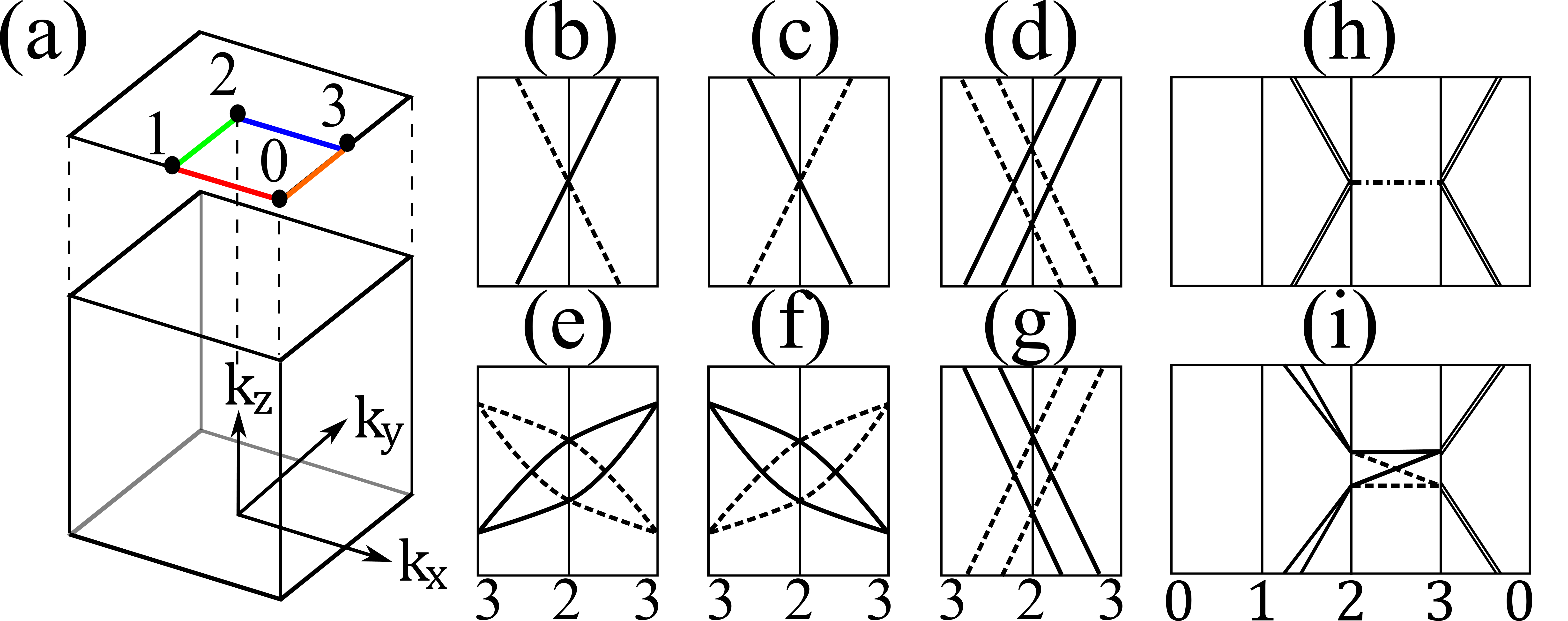}
\caption{ (a) Bottom: Brillouin 3-torus for a glide-symmetric crystal; top:  Brillouin 2-torus corresponding to the glide-symmetric surface. (b-g) Possible surface states on the glide-invariant line $323$; a surface band in the even (odd) representation of glide  is indicated by a solid (dashed) line. (h-i) Surface states on the high-symmetry line $01230$. Bands along $30$ are Kramers-degenerate owing to the composition of time reversal and glide \cite{Hourglass}. In (h), the additional degeneracy (two-fold along 12, four-fold along 23) originates from the alternating-layering construction; these degeneracies may be split by generic perturbations, as illustrated in (i).    }\label{fig:z4surfacestates}
\end{figure}

The square roots of the hourglass-fermion phase have $\ZZZ_4$ invariant $\chi=1,3$, so they must correspond to the $\nu_0 = 1$ phase by Eq.\,(\ref{glideforget}). As can be argued from the quantization of magnetoelectric response \cite{qi2008B,Wilczek_axion}, which persists in the many-body case \cite{essin2009}, the $\nu_0=1$ phase is robust to interactions in the sense that it cannot be connected to the trivial phase by turning on interactions that preserve the many-body gap and the $U(1)$ and $\T$ symmetries. But we know that interactions that preserve glide in addition to $U(1)$ and $\T$ are a subset of those that preserve $U(1)$ and $\T$. If a system cannot be made trivial by turning on interactions that preserve $U(1)$ and $\T$, then it surely cannot be made trivial by turning on interactions that simultaneously preserve glide, $U(1)$, and $\T$, and our argument is complete.

While the minimalist framework did not enter the argument above, it will enter the argument for the robustness of the hourglass-fermion phase, which is a stronger claim and the subject of the next subsection.

\subsection{Robustness of the hourglass-fermion phase\label{subsec:robustness_hourglass_fermion_phase}}

In this subsection, we argue that the hourglass-fermion phase ($\chi=2$) is robust to interactions. More precisely, we argue that it cannot be connected to the trivial phase ($\chi =0$) by turning on interactions that preserve the many-body gap and the glide, $U(1)$, and $\T$ symmetries.

Let us represent the hourglass-fermion phase by a system obtained through the alternating-layer construction. More specifically, let us put copies of a QSH system on all planes of constant $x\in \ZZZ$ and its image under $y\mapsto -y$ on all $x\in \ZZZ+ 1/2$ planes, without turning on inter-plane coupling. To see this represents the hourglass-fermion phase, we recall that a QSH system and its mirror image have identically dispersing 1-dimensional Dirac fermions on the edge. When layered together as described, we obtain two degenerate surface Dirac fermions that do not disperse as functions of $k_x$. This is illustrated in \fig{fig:z4surfacestates}(h), where along the glide-symmetric line 23 we have a four-fold degeneracy originating from two degenerate Dirac points. Glide-symmetric interlayer coupling can only perturb the surface band structure into a connected graph like in \fig{fig:z4surfacestates}(i) \cite{Hourglass}, owing to a combination of the Kramers degeneracy and the monodromy \cite{connectivityMichelZak} of the representation of glide. The resultant connected graph over 0123 has the same topology as the hourglass-fermion phase. Since the bulk gap is maintained throughout the perturbation, the unperturbed system must be in the hourglass-fermion phase. A tight-binding model that demonstrates the construction has been devised by Ezawa \cite{Ezawa_hourglass}.

Next, let us recognize that the robustness of the hourglass-fermion phase to interactions is equivalent to its nontriviality as a 3D SPT phase protected by glide, $U(1)$, and $\T$. A corollary to our general result to be presented in Sec.\,\ref{sec:general_relations} says that given any symmetry $G$, dimension $d$, and nontrivial $(d-1)$-dimensional $G$-protected SPT phase $[a]$, \emph{the $d$-dimensional $\ZZZ \times G$-protected SPT phase obtained from $[a]$ through the alternating-layer construction is trivial if and only if $[a]$ has a square root}. Since the hourglass-fermion phase can be obtained from the QSH phase through the alternating-layer construction, its robustness to interactions now boils down to the absence of a square root of the QSH phase.

To support the last claim, we employ a many-body generalization \cite{qi_spincharge,lee2008,aa2011} of the 2D $\ZZZ_2$ topological invariant $\Delta\in \braces{0,1}$ in Wigner-Dyson class AII. Following the approach of Ref.\ \cite{qi_spincharge}, which is closely related to a preceding pumping formulation \cite{fu2006} that generalizes the well-known Laughlin argument \cite{laughlin1981}, we define $\Delta$ to be the parity of the charge that is pumped toward a flux tube as half a quantum of \emph{spin flux} is threaded. This charge is quantized to be integers even in the many-body case. Moreover, it adds under stacking: if two systems $a$ and $b$ are stacked, then the charge pumped in the stacked system $a+b$ must be a sum of the individual systems. Now it is obvious that not only is a phase with odd $\Delta$ nontrivial, but it also cannot have any square root.

To recapitulate, we have argued for the robustness to interactions of all three nontrivial band insulators in the non-interacting $\ZZZ_4$ classification, by employing only a corollary to our general result. The full power of the general result will be manifest in \s{subsec:sanity_check} when we derive from it the complete classification of 3D SPT phases protected by glide, $U(1)$, and $\T$, which contains $\ZZZ_4$ as a subgroup.

\section{General relations between interacting classifications\label{sec:general_relations}}

From now on we will be dealing with an \emph{arbitrary} symmetry $G$, and SPT phases protected by either $G$ or $G$ combined with a glide symmetry $\ZZZ$. Since our general results apply to both fermionic and bosonic SPT phases, we will often omit such adjectives as ``fermionic" and ``bosonic," with the understanding that $G$ denotes a full symmetry group, which contains fermion parity, in the fermionic case. The Wigner-Dyson class AII, to which hourglass fermions belong, corresponds to the fermionic case and a $G$ that is generated by charge conservation and a time reversal that squares to fermion parity -- it is the unique non-split $U(1)$-extension of $\ZZZ_2$ for the non-trivial action of $\ZZZ_2$ on $U(1)$. We will present the main result of this section, a short exact sequence relating the classification of $G$- and $\ZZZ\times G$-protected SPT phases, in Sec.\,\ref{subsec:SES_classifications}. We will then explore its implications and derive some useful corollaries in Sec.\,\ref{subsec:corollaries}.
Before delving into the results, let us first clarify our terminology.

\subsection{Definition of SPT phases and weakness with respect to glide\label{subsec:SRE_SPT}}

Following Refs.\,\cite{Kitaev_Stony_Brook_2011_SRE_2, Kitaev_Stony_Brook_2013_SRE, Kapustin_Boson, Freed_SRE_iTQFT, Freed_ReflectionPositivity, McGreevy_sSourcery, Xiong}, we shall define $G$-protected SPT phases as $G$-symmetric phases that have ``inverses," in a sense we now make precise.

The key here is that, given any symmetry $G$ and dimension $d$, the stacking operation (see Fig.\,\ref{fig:stacking}) induces a binary operation on the set of \emph{deformation classes} of $d$-dimensional, $G$-symmetric, gapped, local quantum systems 
\footnote{We consider only systems that do not break the symmetry spontaneously.}.
The binary operation has an identity, which is represented by any system with a trivial product ground state \cite{Wen_Definition}. With respect to this identity, we can divide the deformation classes into those that have inverses and those that do not. We shall call the invertible ones $d$-dimensional $G$-protected SPT phases. Consequently, the set of $d$-dimensional $G$-protected SPT phases acquires an abelian group structure under stacking, and we shall denote this abelian group by $\SPT^d\paren{G}$, or $\SPT^d(G,\phi)$ for completeness. If necessary, subscripts can be introduced to distinguish between bosonic and fermionic phases, as in $\SPT^d_b$ or $\SPT^d_f$.

It has been argued that a gapped, local quantum system with on-site symmetry $G$ represents a $G$-protected SPT phase if and only if it has a unique ground state on all manifolds, and that in the 2D case this amounts to the condition of no nontrivial quasiparticle excitations \cite{Kitaev_Stony_Brook_2011_SRE_2, Kitaev_Stony_Brook_2013_SRE, Kapustin_Boson, Freed_SRE_iTQFT, Freed_ReflectionPositivity, McGreevy_sSourcery}. Since this is true of $p+ip$ superconductors \cite{Volovik_p+ip, Read_p+ip, Ivanov_p+ip}, integer quantum Hall systems, the Majorana chain \cite{Majorana_chain}, the $E_8$ model \cite{Kitaev_honeycomb, 2dChiralBosonicSPT, 2dChiralBosonicSPT_erratum, Kitaev_KITP}, etc., such systems should be said to represent SPT phases in our definition.
As demonstrated in Ref.\,\cite{McGreevy_sSourcery} by a worm hole array argument, the inverse of an SPT phase protected by an \emph{on-site} symmetry is given by its orientation-reversed version. That is, if $a$ is a system that represents an SPT phase $[a]$, then the orientation-reversed system $\bar a$ will represent the inverse SPT phase:
\begin{equation}
-[a] = [\bar a].
\end{equation}
The situation with non-on-site symmetries is more involved: with glide, the orientation-reversed version of either square root of the hourglass fermion phase is itself rather than the inverse, for instance.

Another useful notion is that of ``weakness with respect to glide."
Writing $\ZZZ$ for a glide symmetry, we say a $\ZZZ\times G$-protected SPT phase is weak with respect to glide if it becomes trivial under the glide-forgetting map:
\begin{equation}
\beta': \SPT^d\paren{\ZZZ \times G} \fromto \SPT^d\paren{G}. \label{beta}
\end{equation}
We shall denote the set of $d$-dimensional $\ZZZ\times G$-protected SPT phases that are weak with respect to glide by $\wSPT^d(\ZZZ \times G)$, or $\wSPT^d(\ZZZ\times G,\phi)$ for completeness. It is precisely the kernel of $\beta'$
\begin{equation}
\wSPT^d\paren{\ZZZ \times G} \coloneq \kernel \beta',
\end{equation}
which is a subgroup of the abelian group of $d$-dimensional $\ZZZ\times G$-protected SPT phases. Again, subscripts can be introduced to distinguish between bosonic and fermionic phases, as in $\wSPT^d_b$ or $\wSPT^d_f$, if necessary.

We note that there is a different, more traditional definition of SPT phases 
in terms of whether a system can be deformed to a
trivial product state if no symmetry is respected \cite{Wen_Definition}. SPT phases in the traditional sense form a subgroup of the SPT phases in the invertible sense. We have adopted the latter definition because certain SPT phases that are weak with respect to glide can be obtained through the alternating-layer construction only from SPT phases in the invertible sense and not from any SPT phases in the traditional sense; we will exemplify this claim by layering class-A Chern insulators in \s{subsec:classA}.

\subsection{Short exact sequence of classifications\label{subsec:SES_classifications}}

We now present the main result of the section, a short exact sequence that relates the classification of $d$-dimensional $\ZZZ \times G$-protected SPT phases $\SPT^d(\ZZZ \times G)$, the classification of $d$-dimensional $G$-protected SPT phases $\SPT^d(G)$, and the classification of $(d-1)$-dimensional $G$-protected SPT phases $\SPT^{d-1}(G)$, where $G$ is arbitrary and $\ZZZ$ is generated by a glide reflection.

Given any abelian group $A$, we write
\begin{equation}
2A \coloneq \braces{ 2a | a\in A }, \label{2A}
\end{equation}
for the subgroup of $A$ of those elements that have square roots, and we write $A/2A$ for the quotient of $A$ by $2A$. For example, $\ZZZ_n/ 2\ZZZ_n = \ZZZ_{\gcd(n,2)}$, where $\gcd$ stands for greatest common divisor; this even applies to $n= \infty$ if one defines $\gcd(\infty, 2) = 2$.

\begin{widetext}
\begin{prp}[Short exact sequence of classifications]
Assume the Twisted Generalized Cohomology Hypothesis. Let $d$ and $G$ be arbitrary and $\ZZZ$ be generated by a glide reflection. There is
a short exact sequence
\begin{equation}
0 \fromto \SPT^{d-1}(G) / 2\SPT^{d-1}(G) \xfromto{\alpha} \SPT^d\paren{\ZZZ\times G} \xfromto{\beta}  \{ [c] \in \SPT^d\paren{G} \big| 2[c] = 0 \} \fromto 0,\label{SES}
\end{equation}
where $\beta$ is the glide-forgetting map [same as $\beta'$ in Eq.\,(\ref{beta}) but with restricted codomain].
\label{prp:SES_classifications}
\end{prp}
\end{widetext}

\begin{figure}[t]
\centering
\includegraphics[height=1.3in]{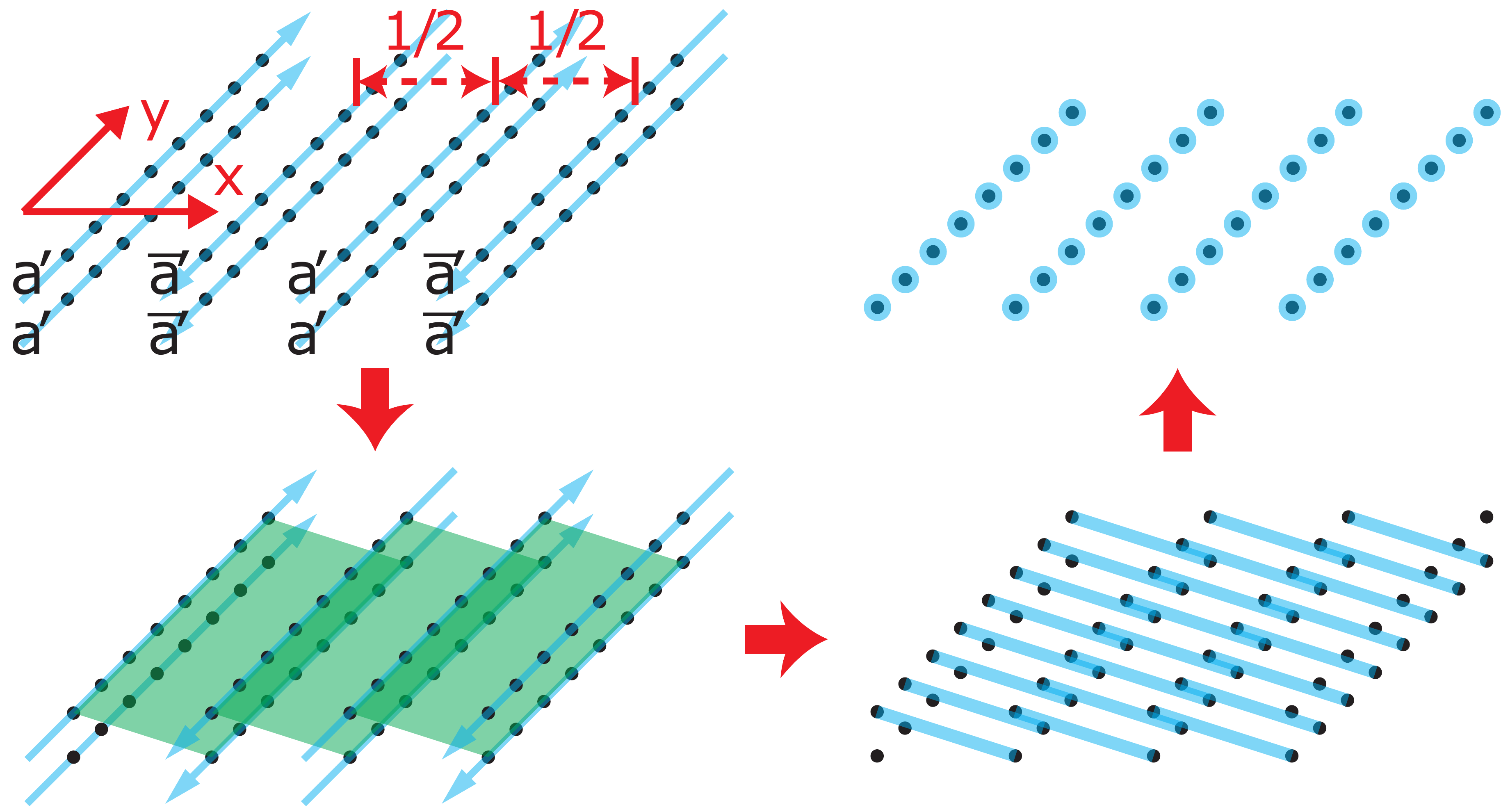}
\caption{Physical justification for the claim $[a] = [a']+[a'] \Rightarrow \alpha'([a]) = 0$, depicted for $d=2$, $(x,y)\mapsto (x+1/2, -y)$. Applying the alternating-layer construction to $a'+a'$ gives a 2D system (upper-left panel). By coupling an $a'$ (or $\overline{a'}$) in each $x\in \ZZZ$ (resp.\,$x\in \ZZZ + 1/2$) layer to an $\overline{a'}$ (resp.\,$a$) in the $x+1/2$ layer (lower-left panel), one can deform the ground state to a tensor product of individual states supported on diagonal pairs of sites (lower-right panel) \cite{Note31}. A redefinition of sites then turns this into a tensor product of individual states supported on single sites, i.e.\,a trivial product state (upper-right panel) \cite{Note41}. All deformations can be chosen to preserve $\ZZZ \times G$ and the gap.}
\label{fig:property_alpha'}
\end{figure}

\begin{proof}
See App.\,\ref{app:proof}.
\end{proof}

While the interpretation of $\beta$ in Proposition \ref{prp:SES_classifications} is clear from the proof in App.\ref{app:proof},
the latter does not address the question as to what $\alpha$ means. Motivated by the discussions in Sec.\,\ref{subsec:alternating_layer_construction_triviality_glide_forgetting}, we shall posit that $\alpha$ is given by the alternating-layer construction introduced earlier (see Fig.\,\ref{fig:alternating_layer_construction}).
More precisely, the alternating-layer construction defines a map
\begin{equation}
\alpha': \SPT^{d-1}(G) \fromto \SPT^d\paren{\ZZZ \times G}, \label{alpha'}
\end{equation}
with domain the abelian group of $(d-1)$-dimensional $G$-protected SPT phases.
By the physical argument in Fig.\,\ref{fig:property_alpha'}, we must have $\alpha'([a]) = 0$
whenever there exists an $[a']\in \SPT^{d-1}(G)$ such that $[a] = [a'] + [a']$. This means $\alpha'$ can effectively be defined on the quotient $\SPT^{d-1}(G)/ 2\SPT^{d-1}(G)$, and we shall identify the induced map
\begin{equation}
\alpha: \SPT^{d-1}(G) / 2 \SPT^{d-1}(G) \fromto \SPT^d\paren{\ZZZ \times G} \label{alpha}
\end{equation}
as the $\alpha$ that appears in Proposition \ref{prp:SES_classifications}.
This interpretation of $\alpha$ is further supported by the various examples in Sec.\,\ref{sec:applications} below and Ref.\,\cite{Lu_sgSPT}.%
\phantom{
\setcounter{footnote}{30}
\footnote{In general it is not necessary, or even possible, to go through this intermediate step. In the discussion of SPT phases, one always considers stable equivalence. That is, one allows for change of Hilbert spaces through, for instance, the introduction of ancillary lattices or a blocking of lattice sites \cite{Wen_Boson, Cirac}. The deformation of a composite system into a system whose ground state is a tensor product may already involve such changes, in which case the intermediate picture of a tensor product of states supported on pairs of sites will no longer be accurate.}
\setcounter{footnote}{40}
\footnote{More precisely, the transformation of states supported on pairs of sites into states supported on single sites involves first an enlargement of the Hilbert space and then a transfer of states in one sector of the enlarged Hilbert space to another \cite{Wen_Boson, Cirac}.}
}

We will see in Sec.\,\ref{sec:physical_picture} that twice the glide reflection being orientation-preserving is closely related to the factors of 2 appearing in Proposition \ref{prp:SES_classifications}. The reader may have realized that $\SPT^{d-1}(G) / 2\SPT^{d-1}(G)$ and $\{ [c] \in \SPT^d\paren{G} | 2[c] = 0 \}$ can be expressed using the extension and torsion functors, as $\Ext^1( \ZZZ_2, \SPT^{d-1}(G) )$ and $\Tor_1 (\ZZZ_2, \SPT^d\paren{G})$, respectively, both of which are contravariant as they should be \cite{Xiong}.

\subsection{Implications of the short exact sequence \label{subsec:corollaries}}

Let us explore the implications of Proposition \ref{prp:SES_classifications} and derive some useful corollaries from it.

First, the exactness of sequence (\ref{SES}) implies that $\image \alpha = \kernel \beta$ (this can be equivalently stated as $\image \alpha' = \kernel \beta'$ since, by definition, $\image \alpha = \image \alpha'$ and $\kernel \beta = \kernel \beta'$), which reproduces the result in Ref.\,\cite{Lu_sgSPT} that
\begin{cor}
A $\ZZZ\times G$-protected SPT phase is weak with respect to glide if and only if it can be obtained through the alternating-layer construction.
\end{cor}

\noindent Furthermore, since every element of an abelian group of the form $A/2A$ is either trivial or of order 2, any $\ZZZ \times G$-protected SPT phase that is weak with respect to glide -- and hence obtainable from an element of $\SPT^{d-1}(G) / 2 \SPT^{d-1}(G)$ -- must also be either trivial or of order 2. In fact, a necessary and sufficient condition for such a phase to be trivial (resp. has order 2) is that the $(d-1)$-dimensional $G$-protected SPT phase it comes from has a square root (resp. has no square root)
\footnote{
While two $(d-1)$-dimensional $G$-protected SPT phases may lead to the same $d$-dimensional $\ZZZ \times G$-protected SPT phase, this if-and-only-if condition is unambiguous because any two such phases must both have square roots or both have no square roots.}.
This follows from the exactness of sequence (\ref{SES}), which implies $\alpha$ is injective. 

The above necessary and sufficient condition allows us to classify $d$-dimensional $\ZZZ \times G$-protected SPT phase that are weak with respect to glide by classifying $(d-1)$-dimensional $G$-protected SPT phases instead:
\begin{cor}
\label{cor:classification_SPT_weak_wrt_glide}
There is an isomorphism
\begin{equation}
\wSPT^d\paren{\ZZZ \times G} \isomorphic \SPT^{d-1}(G) / 2\SPT^{d-1}(G). \label{GSPT}
\end{equation}
\end{cor}

\noindent This isomorphism was conjectured in Ref.\,\cite{Lu_sgSPT} for on-site $G$, based on studies of a number of fermionic and bosonic examples in the $d=3$ case. Unfortunately, the anomalous surface topological order argument used therein does not generalize to all dimensions. The minimalist framework allows us to confirm their conjecture in the general case -- in all dimensions and for all symmetries $G$, which do not even have to act in an on-site fashion.

We are concerned with all $\ZZZ\times G$-protected SPT phases, not just those that are weak with respect to glide. We know that by forgetting glide each $d$-dimensional $\ZZZ\times G$-protected SPT phase can be viewed as a $d$-dimensional $G$-protected SPT phase, but can all $d$-dimensional $G$-protected SPT phases be obtained this way? In other words, are all $d$-dimensional $G$-protected SPT phases compatible with glide? Our result indicates that the answer is in general no. This is because the exactness of sequence (\ref{SES}) implies $\beta$ is surjective, and inspecting the third term (from the left, excluding the initial 0) of sequence (\ref{SES}) one sees that
\begin{cor}[Compatibility with glide]
A necessary and sufficient condition for a $d$-dimensional $G$-protected SPT phase to be compatible with glide is that it squares to the trivial phase.
\end{cor}

In Corollary \ref{cor:classification_SPT_weak_wrt_glide} we saw that the classification of $d$-dimensional $\ZZZ \times G$-protected SPT phases \emph{that are weak with respect to glide} can be obtained from the classification of $(d-1)$-dimensional $G$-protected SPT phases. We now demonstrate that the classification of \emph{all} $d$-dimensional $\ZZZ \times G$-protected SPT phases is severely constrained, if not completely determined, by the classification of $(d-1)$- and $d$-dimensional $G$-protected SPT phases. Indeed, we recognize that the task of determining the second term of a short exact sequence of abelian groups from the first and third terms is nothing but an abelian group extension problem. It is well-known that abelian group extensions $0 \fromto A \fromto B \fromto C \fromto 0$ of $C$ by $A$ are classified, with respect to a suitable notion of equivalence, by the subgroup $H^2_{\rm sym}\paren{C;A}$ of $H^2\paren{C; A}$ of symmetric group cohomology classes.
To illustrate how $A$ and $C$ constrain $B$, let us take $A= C = \ZZZ_2$. In this case, $H^2_{\rm sym}\paren{\ZZZ_2;\ZZZ_2} \isomorphic \ZZZ_2$. The trivial and nontrivial elements of $H^2_{\rm sym}\paren{\ZZZ_2;\ZZZ_2}$ correspond to $B = \ZZZ_2 \oplus \ZZZ_2$ and $\ZZZ_4$, respectively, which are the only solutions to the abelian group extension problem.

Inspecting (\ref{SES}), we note that its first and third terms are such that their nontrivial elements all have order 2. Consequently, all nontrivial elements of the second term must have order 2 or 4. More precisely, we have
\begin{cor}[Quad-chotomy of phases]
\label{cor:quad-chotomy}
Each $\ZZZ\times G$-protected SPT phase is exactly one of the following:
\begin{enumerate}[(i)]
\item the unique trivial phase;

\item a nontrivial phase of order 2 that is weak with respect to glide;

\item a nontrivial phase of order 2 that is not weak with respect to glide;

\item a nontrivial phase of order 4 that is not weak with respect to glide per se but whose square is of type (ii).
\end{enumerate}
\end{cor}

\begin{proof}
See App.\,\ref{app:proof_corollaries}.
\end{proof}

In particular, this means that a nontrivial $\ZZZ\times G$-protected SPT phase that is weak with respect to glide can sometimes have square roots, and that such square roots, if exist, are never weak with respect to glide. On the other hand, a nontrivial $\ZZZ\times G$-protected SPT phase that is \emph{not} weak with respect to glide can \emph{never} have square roots. Without Proposition \ref{prp:SES_classifications}, these results would not have been obvious.

From the perspective of classification, it would be nice to have a statement about the explicit form of $\SPT^d(\ZZZ \times G)$. In App.\,\ref{app:proof_corollaries}, we prove that Corollary \ref{cor:quad-chotomy}, together with the fact that SPT phases form an abelian group, implies that $\SPT^d(\ZZZ \times G)$ can be written as a direct sum of $\ZZZ_4$'s and $\ZZZ_2$'s:

\begin{cor}[Direct-sum decomposition]
\label{cor:direct_sum_decomposition}
There is a direct sum decomposition,
\begin{equation}
\SPT^d(\ZZZ \times G) \isomorphic \paren{\bigoplus_i \ZZZ_4} \oplus \paren{\bigoplus_j \ZZZ_2 }.\la{directsumdecomp}
\end{equation}
$\ZZZ \times G$-protected SPT phases that correspond to 1 or 3 of any $\ZZZ_4=\braces{0,1,2,3}$ summand are never weak with respect to glide, whereas those that correspond to $2\in \ZZZ_4$ are always weak with respect to glide. The nontrivial element of a $\ZZZ_2$ summand, on the other hand, may or may not be weak with respect to glide.
\end{cor}

\section{Physical intuition behind Proposition \ref{prp:SES_classifications}\label{sec:physical_picture}}

Proposition \ref{prp:SES_classifications} was derived mathematically from the Twisted Generalized Cohomology Hypothesis. Here we offer a complementary physical perspective. We will break down Proposition \ref{prp:SES_classifications} into 6 individual statements and explain them \emph{physically} where possible:
\begin{enumerate}[(i)]

\item $\image \alpha' \subset \kernel \beta'$;

\item $\image \alpha' \supset \kernel \beta'$;

\item $\alpha$ is well-defined;

\item $\alpha$ is injective;

\item $\image \beta' \subset \{ [c] \in \SPT^d\paren{G} \big| 2[c] = 0 \}$;

\item $\image \beta' \supset \{ [c] \in \SPT^d\paren{G} \big| 2[c] = 0 \}$.
\end{enumerate}
We will discuss (i) and (ii) in Sec.\,\ref{subsec:alternating_layer_construction_triviality_glide_forgetting}, (iii) and (iv) in Sec.\,\ref{subsec:square_root_triviality}, and (v) and (vi) in Sec.\,\ref{subsec:compatibility_glide_2_torsion}. Of the six statements, (i)(ii)(iii)(v)(vi) admit obvious physical explanation, whereas (iv) can be justified by physical examples. Although the proof of Proposition \ref{prp:SES_classifications} was rigorous and the Twisted Generalized Cohomology Hypothesis can largely be justified on independent grounds, it is but reassuring that Proposition \ref{prp:SES_classifications} is consistent with one's physical intuition.

\subsection{Alternating-layer construction and triviality under glide forgetting\label{subsec:alternating_layer_construction_triviality_glide_forgetting}}

\begin{figure}[t]
\centering
\includegraphics[width=3.3in]{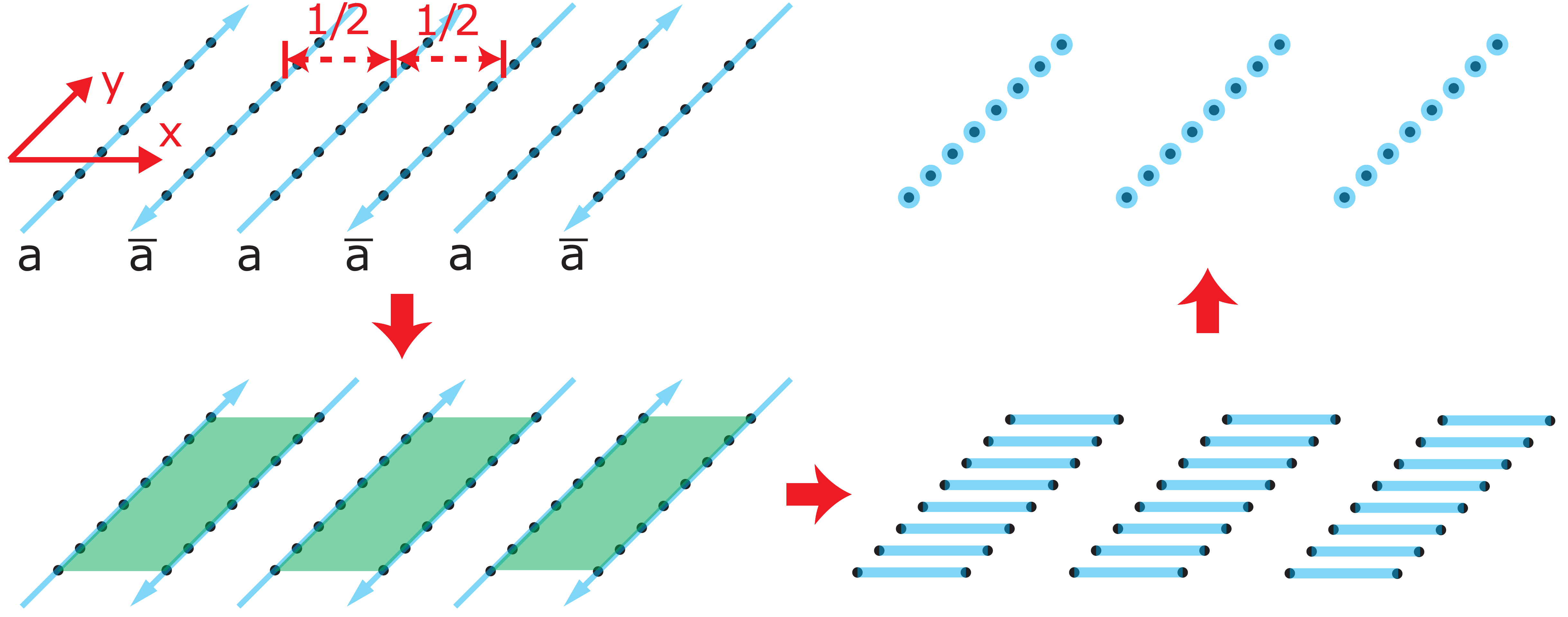}
\caption{Physical justification for the claim $\image \alpha' \subset \kernel \beta'$, depicted for $d=2$, $(x,y)\mapsto (x+1/2, -y)$. Applying the alternating-layer construction to $a$ gives a 2D system (upper-left panel). By coupling each $x\in \ZZZ$ layer to the $x+1/2$ layer (lower-left panel), one can deform the 2D system so as to have a ground state that is the tensor product of individual states supported on pairs of sites (lower-right panel) \cite{Note31}. A blocking procedure then turns the latter into a tensor product of individual states supported on single sites (upper-right panel). All deformations can be chosen to preserve $G$ and the gap.}
\label{fig:image_alpha_subset_kernel_beta}
\end{figure}

Given our interpretation of $\alpha'$ as the alternating-layer construction, $\image \alpha' {\subset} \kernel \beta'$ amounts to saying that \emph{the alternating-layer construction always produces $\ZZZ\times G$-protected SPT phases that are weak with respect to glide}. In other words, given a system $a$ representing a $(d-1)$-dimensional $G$-protected SPT phase, the $d$-dimensional system obtained from $a$ through the alternating-layer construction can always be trivialized when the glide symmetry constraint is relaxed. Indeed, given such a $d$-dimensional system, one can simply pair up neighboring layers and deform the pairs into trivial systems, as illustrated in Fig.\,\ref{fig:image_alpha_subset_kernel_beta}.

The converse, $\kernel \beta' \subset \image \alpha'$, says that \emph{all $\ZZZ\times G$-protected SPT phases that are weak with respect to glide can be obtained through the alternating-layer construction}. That is, if a system $b$ representing a $\paren{\ZZZ \times G}$-protected SPT phase can be trivialized when the glide symmetry constraint is relaxed, then it can be deformed to a system obtained from the alternating-layer construction while preserving the glide symmetry. Indeed, an argument involving applying a symmetric, finite-depth quantum circuit to subregions of a glide-symmetric system has been devised in Ref.\,\cite{Lu_sgSPT} to justify this claim, assuming the lattice period is large compared to the correlation length.

We can view the above as physically motivating our identification of $\alpha'$ as the alternating-layer-construction map in the first place. We will soon be delighted to find out that this interpretation is consistent with the other statements as well.

\subsection{Square root in $(d-1)$ dimensions and triviality in $d$ dimensions \label{subsec:square_root_triviality}}

The well-definedness of $\alpha$ says, \emph{given a $(d-1)$-dimensional $G$-protected SPT phase $[a]$, that the $d$-dimensional $\ZZZ\times G$-protected SPT phase obtained from it through the alternating-layer construction is trivial whenever $[a]$ has a square root}. As mentioned in Sec.\,\ref{subsec:SES_classifications}, we can justify this claim using the physical argument in Fig.\,\ref{fig:property_alpha'}. Ref.\,\cite{Lu_sgSPT} has also given an equivalent argument.

On the other hand, the injectivity of $\alpha$ says that \emph{the $d$-dimensional $\ZZZ\times G$-protected SPT phase obtained from $[a]$ through the alternating-layer construction is nontrivial whenever $[a]$ has no square root}.
Physically, this has been shown to be the case for a number of bosonic and fermionic systems for $d=3$ using the $K$-matrix construction \cite{Lu_sgSPT}. In general, one can attempt a construction of bulk invariants for the $d$-dimensional system in question, but a universal strategy that works for all $d$ seems lacking.

\subsection{Compatibility with glide and $\ZZZ_2$ torsion \label{subsec:compatibility_glide_2_torsion}}

Finally, let us argue that
\begin{equation}
\image \beta' \subset \{ [c] \in \SPT^d\paren{G} \big| 2[c] = 0 \},
\end{equation}
which amounts to saying that \emph{if a $G$-protected SPT phase has a glide-symmetric representative} (i.e.\,is compatible with glide), \emph{then it must square to the trivial phase} (i.e.\,belong to the $\ZZZ_2$ torsion subgroup). In other words, given a system $b$ representing a $d$-dimensional $\ZZZ \times G$-protected SPT phase, the stacked system $b+b$ can always be trivialized when the glide symmetry constraint is relaxed.

Let us begin by considering the stacked system $b + \bar b$, where $\bar b$ is the mirror image of $b$ under $y \mapsto -y$ [Fig.\,\ref{fig:image_beta}(a)]. When glide is relaxed,
$\bar b$ serves as the inverse of $b$. This means one can deform $b + \bar b$ to a trivial product state [Fig.\,\ref{fig:image_beta}(b)]. Now we translate $\bar b$ by $1/2$ in the $x$-direction. Then the stacked system $b + (\mbox{translated }\bar b)$ [Fig.\,\ref{fig:image_beta}(c)] can also be deformed to a tensor product state, or more precisely, a tensor product of individual states that are supported on diagonal pairs of sites [Fig.\,\ref{fig:image_beta}(d)] \cite{Note31}. A redefinition of sites then turns the latter into a tensor product of individual states supported on single sites, that is, into a trivial product state [Fig.\,\ref{fig:image_beta}(f)] \cite{Note41}. To see that $b+b$ can be deformed to a trivial product state, we simply need to note that, being glide-symmetric, $b$ is the same as the translated $\bar b$.

\begin{figure}[t]
\centering
\includegraphics[width=3.3in]{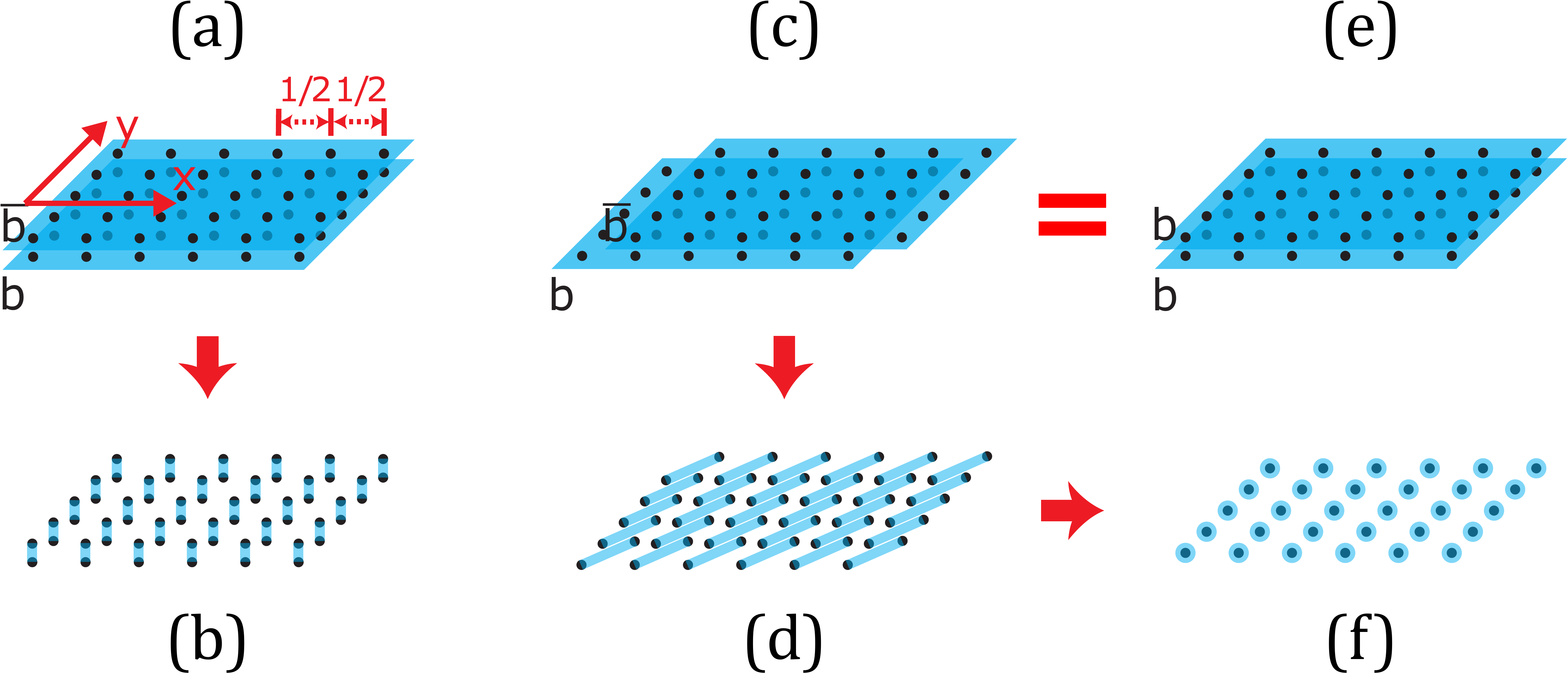}
\caption{Physical justification for the claim $\image \beta' \subset \{ [c] \in \SPT^d\paren{G} \big| 2[c] = 0 \}$, depicted for $d=2$, $(x,y)\mapsto (x+1/2, -y)$.}
\label{fig:image_beta}
\end{figure}

\begin{figure}[t]
\centering
\includegraphics[width=2.3in]{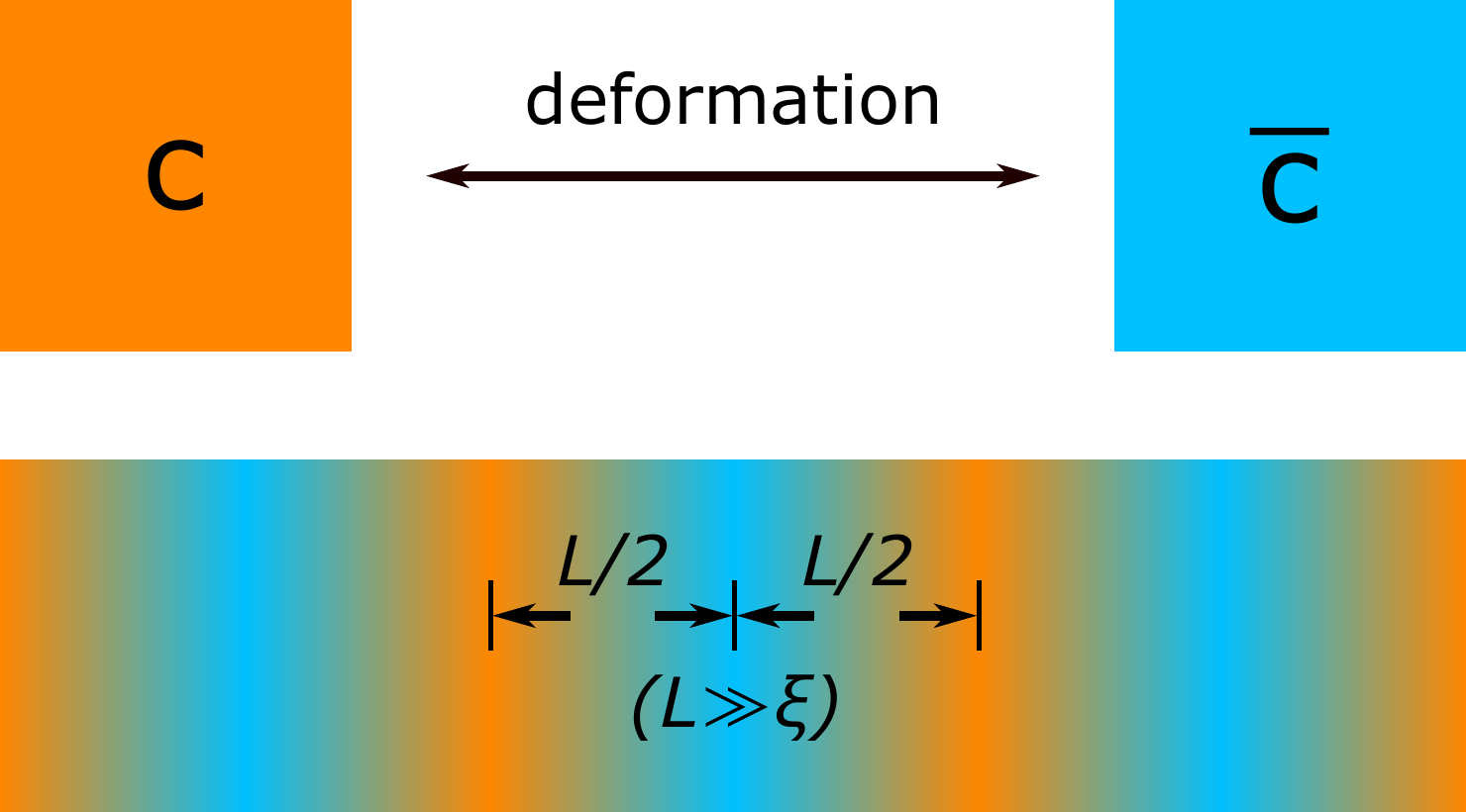}
\caption{Physical justification for the claim $\{ [c] \in \SPT^d\paren{G} \big| 2[c] = 0 \} \subset \image \beta'$, depicted for $d=2$.}
\label{fig:lastresult}
\end{figure}

The converse,
\begin{equation}
\{ [c] \in \SPT^d\paren{G} \big| 2[c] = 0 \} \subset \image \beta', \la{lastresult}
\end{equation}
says that \emph{if a $d$-dimensional $G$-protected SPT phase squares to the trivial phase, then it must have a glide-symmetric representative.} As a quick argument for this, we appeal to the empirical beliefs that (a) an SPT phase that squares to the trivial phase has a reflection-symmetric representative, and that (b) any SPT has a translation-invariant representative \cite{Xiong}, which are consistent with known examples. A case in point for (a) is the reflection-symmetric topological insulator Bi$_2$Se$_3$ ($\nu_0=1$ in class AII); (b) is exemplified by all experimentally realized band topological insulators. Now, suppose (a) and (b) can be compatibly realized in the same system, with the reflection axis ($x_2 \mapsto -x_2$) orthogonal to at least one translation direction ($x_1 \mapsto x_1 + 1/2$). Composing the two transformations, we see that the system is also invariant under the glide symmetry $(x_1,x_2,\ldots)\mapsto (x_1+1/2,-x_2,\ldots)$.

An alternative argument which does not depend on the belief (a) above is this. Suppose a system $c$ represents a $G$-protected SPT phase $[c]$ that squares to the trivial phase. The condition $2[c]=0$ is equivalent to the condition $[c] = -[c]$, or $[c] = \brackets{\bar c}$, where $\bar c$ is the orientation-reversed (say $x_2 \mapsto -x_2$) version of $c$. The last expression means that $c$ can be deformed to $\bar c$ without closing the gap or breaking the symmetry (see the upper panel of Fig.\,\ref{fig:lastresult}). Let $\hat H(\lambda)$ be a family of translation-invariant Hamiltonians parameterized by $\lambda\in [0,1]$ that represents this deformation. Since $\lambda$ is a compact parameter ($[0,1]$ being closed and bounded), we expect the correlation length of $\hat H(\lambda)$ to be uniformly bounded by some finite $\xi$ \cite{hastings2006spectral}. Being translation-invariant, each $\hat H(\lambda)$ is a sum of the form
\begin{equation}
\hat H(\lambda) = \sum_{\boldsymbol x} \sum_i g_i(\lambda) \hat O^i_{\boldsymbol x},
\end{equation}
where $\boldsymbol x = \paren{x_1, x_2, \ldots}$ runs over all lattice points, the operators $\hat O^i$ have compact supports (of radii $r_i$), $\hat O^i_{\boldsymbol x}$ denotes the operator $\hat O^i$ centered at $\boldsymbol x$, the coupling constants $g_i(\lambda)$ depend on $\lambda$, and $g_i(\lambda)$ decay exponentially with $r_i$. Now, we construct a new Hamiltonian $\hat H'$ that modulates spatially at a scale $L$ much larger than $\xi$. This can be achieved by letting $\lambda$ vary with one of the coordinates, say $x_1$; for instance, we can set
\begin{equation}
\lambda = \frac{x_1}{L/2}
\end{equation}
for $x_1 \in \brackets{0,L/2}$. In the neighborhood of $x_1 = 0$ and $L/2$, the Hamiltonian $\hat H'$ will coincide with $\hat H(0)$ and $\hat H(1)$, respectively. This defines $\hat H'$ only in the strip $x_1 \in \brackets{0,L/2}$, but since $\hat H(1)$ is related to $\hat H(0)$ by $x_2 \mapsto -x_2$, we can place the reversed strip on $x_1 \in \brackets{L/2,L}$ and glue the two strips together. Iterating this process {\it ad infinitum} to create a superlattice, we will arrive at a Hamiltonian $\hat H'$ that is explicitly invariant under the glide transformation $\paren{x_1, x_2, \ldots} \mapsto \paren{x_1+L/2, -x_2, \ldots}$; see the lower panel of Fig.\,\ref{fig:lastresult}. Due to the separation of scale $L \gg \xi$, we expect $\hat H'$ to be gapped. To see that $\hat H'$ represents the SPT phase $[c]$, we note that in the neighborhood of any $x_1 \in [0,L/2]$ (resp.\,$[L/2,L]$), there is some $\lambda$ for which $\hat H'$ is locally indistinguishable from $\hat H(\lambda)$ (resp.\,its orientation-reversed version), which represents $[c]$.

An explicit formula for $\hat H'$ can be given. Let $\hat M$ be the operator that implements the orientation-reversal $x_2 \mapsto -x_2$. Then we can write
\begin{equation}
\hat H' = \sum_{x_1} \hat H'_{x_1},
\end{equation}
where
\begin{equation}
\hat H'_{x_1} = \sum_{x_2, x_3, \ldots} \sum_i g_i\paren{\frac{x_1-nL}{L/2}} \hat O^i_{\boldsymbol x}
\end{equation}
for $x_1 \in \brackets{nL, (n+1/2)L}$, and
\begin{equation}
\hat H'_{x_1} = \sum_{x_2, x_3, \ldots} \sum_i g_i\paren{\frac{x_1-(n+1/2)L}{L/2}} \hat M \hat O^i_{\boldsymbol x} \hat M^{-1}
\end{equation}
for $x_1 \in \brackets{(n+1/2)L,(n+1)L}$. Here $n$ takes values in the integers.

\section{Applications to fermionic SPT phases in classes A and AII\label{sec:applications}}

In this section, we will demonstrate that the predictions of Proposition \ref{prp:SES_classifications} are consistent with existing literature on the classification of free-fermion phases and their robustness to interactions. More importantly, we will use Proposition \ref{prp:SES_classifications} to deduce the putative complete classifications of fermionic SPT phases with glide from proposed complete classifications of fermionic SPT phases without glide. The latter is an abelian group extension problem, where knowing the first and third terms $A$ and $C$ of a short exact sequence of abelian groups,
\begin{equation}
0 \fromto A \fromto \mbox{?} \fromto C \fromto 0,
\end{equation}
one has to determine the second. For definiteness, we will first focus on $d=3$ and $G=U(1)$ (charge conservation only, Wigner-Dyson class A). Then, we will re-examine the symmetry class of the hourglass-fermion phase, where $d=3$ and $G$ is generated by $U(1)$ and $\T$ where $\T$ squares to fermion parity (charge conservation and time reversal, Wigner-Dyson class AII).

\subsection{Wigner-Dyson class A\label{subsec:classA}}

Let us set $d=3$ and $G=U(1)$, which corresponds to Wigner-Dyson class A. 2D free-fermion phases in this symmetry class are classified by the first Chern number ($C_1{\in} \ZZZ$), which is defined over the Brillouin torus \cite{TKNN} but can be generalized to the interacting or disordered case by considering the torus of twisted boundary conditions instead \cite{Niu_twistedBC}. Being robust to interactions and disorder and admitting no square root, a phase with odd Chern number represents a nontrivial element in the first term of sequence (\ref{SES}). This phase may be layered in an alternating fashion to form a 3D phase respecting an additional glide symmetry $\ZZZ$.

The non-interacting, clean limit \footnote{By a clean limit, we mean a system that has discrete translational symmetry in three independent directions.} of the resultant 3D phase was independently studied in Refs.\ \cite{unpinned} and \cite{Shiozaki2015}. The surface states have a characteristic connectivity over the surface Brillouin torus illustrated in \fig{fig:z2surfacestates}; these surface states have been described as carrying a M\"obius twist \cite{Shiozaki2015}, so we shall refer to this phase as the M\"obius-twist phase. It was concluded in both references that the non-interacting classification (class A \emph{with glide}, 3D) is $\ZZZ_2$, and a topological invariant was proposed ($\kappa\in \{0,1\}$) to distinguish the two phases.

\begin{figure}[t]
\centering
\includegraphics[width=7 cm]{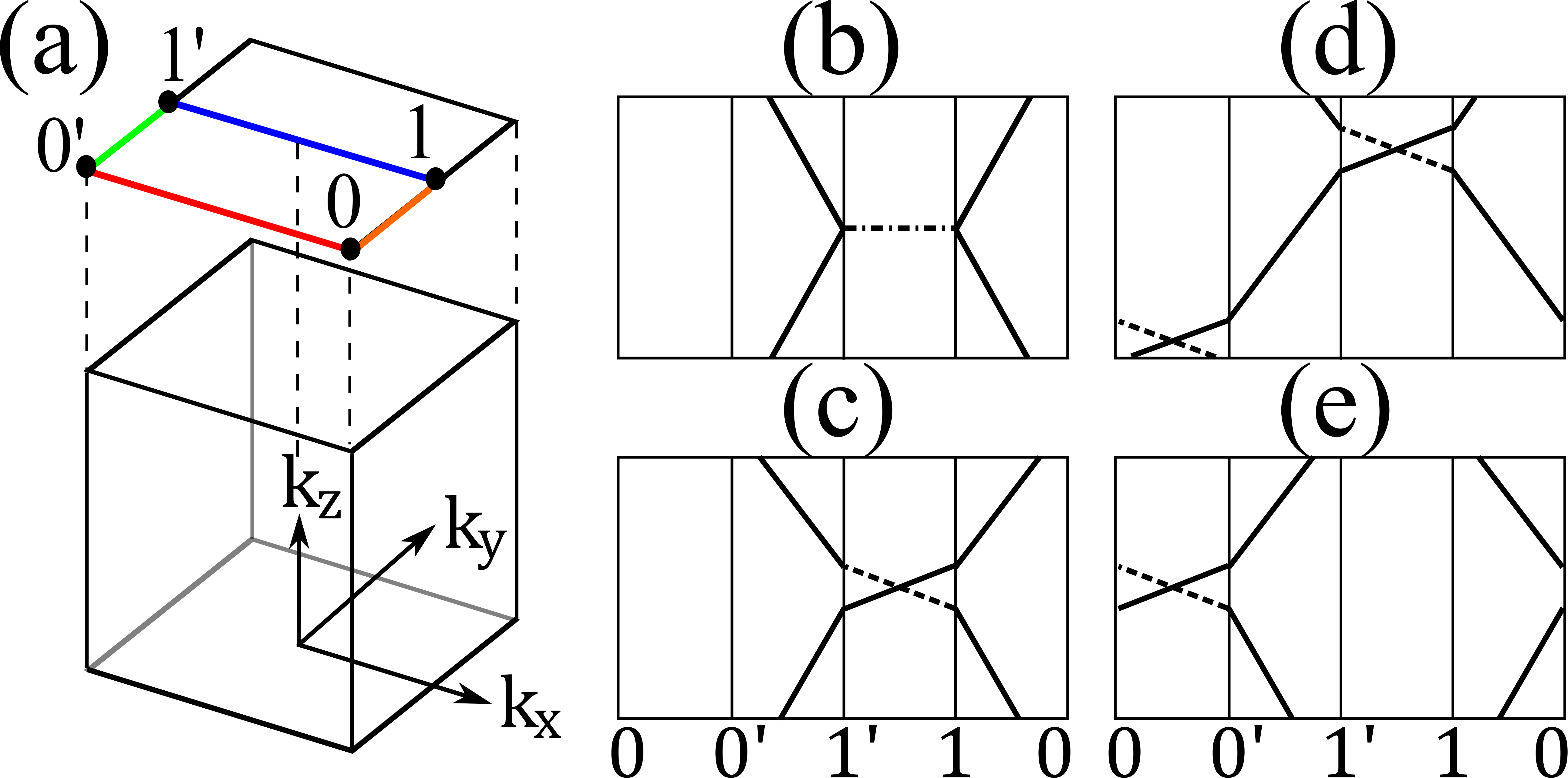}
\caption{ (a) Bottom: Brillouin 3-torus for a glide-symmetric crystal; top:  Brillouin 2-torus corresponding to the glide-symmetric surface. (b-e) Surface states with a Mobius twist; a surface band in the even (odd) representation of glide  is indicated by a solid (dashed) line. (b-e) are representatives of the same phase, i.e., they are connected by symmetric deformations of the Hamiltonian that preserve the bulk gap. In (b), the solid-dashed line indicates a doubly-degenerate band  originating from the alternating-layer construction; this degeneracy may be split by generic perturbations, as illustrated in (c).}\label{fig:z2surfacestates}
\end{figure}

That the M\"obius-twist phase ($\kappa=1$) can be obtained from the alternating-layer construction as above is especially evident in the non-interacting limit, where we can utilize the connectivity of surface states as an argument, in conjunction with the bulk-boundary correspondence \cite{Cohomological}. The following may be viewed as the class-A analog of the argument presented in  \s{subsec:robustness_hourglass_fermion_phase} for class AII.
Let us begin with the $C_1{=}1$ phase with a single edge chiral mode; the mirror image of this Chern phase has $C_1{=}{-}1$ and a single edge chiral mode with opposite velocity -- an ``anti-chiral" mode for short. When layered together in the $x$-direction with vanishing interlayer coupling, we obtain a superposition of chiral and anti-chiral modes (in the shape of an X) that do not disperse with $k_x$, as illustrated in \fig{fig:z2surfacestates}(b). Note in particular the two-fold energy degeneracy along the glide-invariant line 1'1, which originates from the intersection of chiral and anti-chiral modes. If we perturb the system with a glide-symmetric interlayer coupling, the degenerate two-band subspace is bound to split into a connected graph (in the shape of a M\"obius twist) over the glide-invariant line as in \fig{fig:z2surfacestates}(c), owing to the monodromy \cite{connectivityMichelZak} of the representation of glide. The topology of the graph over $1'100'$ then confirms that the system is characterized by $\kappa{=}1$.

One implication of our short exact sequence that goes beyond the aforementioned non-interacting works is that the 3D M\"obius-twist phase ($\kappa=1$) is robust to interactions. More precisely, it cannot be connected to the trivial phase ($\kappa=0$) by turning on interactions that preserve the many-body gap and $\ZZZ \times{ U(1)}$ symmetry. Indeed, we can utilize the same corollary as quoted in \s{subsec:robustness_hourglass_fermion_phase} and recognize that an insulator with odd Chern number admits no square root. Equivalently, we can view this as a direct consequence of the injectivity of the map $\alpha$ in Proposition \ref{prp:SES_classifications}.

An independent argument for the robustness of the M\"obius-twist phase under $\ZZZ \times U(1)$ may be obtained from the quantized magnetoelectric \cite{essin2009} bulk response. In this case, the quantization results from the glide symmetry, which maps the axion angle from $\theta \mapsto -\theta$ \footnote{One way to rationalize this is to apply the pseudo-scalar transformation behavior of $\bE\cdot \bB$.}. Since $\theta$ is defined modulo $2\pi$ \cite{Wilczek_axion}, it is fixed to $0$ or $\pi$, for $\kappa=0$ or $1$, respectively. Another independent argument for the robustness is that the glide-symmetric surface of the M\"obius-twist phase allows for an anomalous topological order of the T-Pfaffian type, which cannot exist in pure 2D glide-symmetric systems \cite{Lu_sgSPT}.

When glide is forgotten, the non-interacting 3D classification in class A is trivial (i.e. there is only one phase). Hence the  M\"obius-twist phase is weak with respect to glide. We may also argue for its weakness by noting that, without glide, there is no obstruction to coupling and trivializing adjacent layers with opposite $C_1$. As a $\ZZZ \times U(1)$-protected SPT phase, the M\"obius-twist phase has order 2 because two copies of itself have $\kappa = 1+1 \equiv 0$. Thus the M\"obius-twist phase falls into category (ii) of Corollary \ref{cor:quad-chotomy}.%
\phantom{\setcounter{footnote}{50}\footnote{By an interaction-enabled fermionic SPT phase, we mean a phase that does not have a free-fermion representative.}}

Let us now include interaction-enabled \cite{Note51} fermionic SPT phases and demonstrate how the relation between classifications, as encapsulated in Proposition \ref{prp:SES_classifications}, can help us pin down the complete classification of 3D $\ZZZ \times U(1)$-protected fermionic SPT phases. It is believed, without glide, that the complete classifications of 2 and 3D $U(1)$-protected fermionic SPT phases are 
\begin{eqnarray}
\SPT^2_f\paren{  U(1) } &\isomorphic& \ZZZ{\oplus}\ZZZ, \\
\SPT^3_f\paren{  U(1) } &\isomorphic& 0,
\end{eqnarray}
respectively, where the first $\ZZZ$ is generated by the $C_1{=}1$ phase and the second $\ZZZ$ by the neutral $E_8$ phase \cite{Kitaev_honeycomb, 2dChiralBosonicSPT, 2dChiralBosonicSPT_erratum, Kitaev_KITP}. Inserting these into the short exact sequence of Proposition \ref{prp:SES_classifications}, we obtain an abelian group extension problem:
\e{0 \rightarrow \ZZZ_2 \oplus \ZZZ_2 \rightarrow\; ? \rightarrow 0 \rightarrow 0.\label{groupextension_A}
}
It is an elementary property of group extension that the group extension of the trivial group by any other group is unique. More generally, if $A=0$ or $C=0$ in a short exact sequence $0 \fromto A \fromto B \fromto C \fromto 0$, then $B\isomorphic C$ or $B\isomorphic A$, respectively. Either way, we conclude that there is a unique solution to Eq.\,(\ref{groupextension_A}), and the complete classification of 3D $\ZZZ \times U(1)$-protected fermionic SPT phases is
\e{\SPT^3_f\paren{\ZZZ {\times}  U(1) }\isomorphic \ZZZ_2 \oplus \ZZZ_2,\la{sptclassA}}
which is consistent with Corollaries \ref{cor:classification_SPT_weak_wrt_glide} and \ref{cor:direct_sum_decomposition}. This result goes beyond the previous work Ref.\,\cite{Lu_sgSPT} in that Ref.\,\cite{Lu_sgSPT} only classified SPT phases that are weak with respect to glide. Our result indicates that, in this case, the ``weak classification" is complete. In the next subsection, we will investigate a case where the weak classification is not complete. We will see that the complete classification can still be determined through our short exact sequence.

\subsection{Wigner-Dyson class AII\label{subsec:sanity_check}}

Let us set $d=3$ and $G$ to be generated by $U(1)$ and $\T$ where $\T$ squares to fermion parity, which corresponds to Wigner-Dyson class AII. As a group, $G$ is the unique non-split $U(1)$-extension of $\ZZZ_2$ for the non-trivial action of $\ZZZ_2$ on $U(1)$.

As mentioned in Sec.\,\ref{sec:hourglass_fermions}, the free-fermion classification in this symmetry class is $\ZZZ_2$ without glide and $\ZZZ_4$ with glide. The hourglass-fermion phase has a $\ZZZ_4$ index $\chi = 2$ and 3D $\ZZZ_2$ index $\nu_0 = 0$, so it represents a $\ZZZ \times G$-protected SPT phase that is weak with respect to glide. It is still a nontrivial SPT phase, though, by the discussion in Sec.\,\ref{sec:hourglass_fermions}. We commented in Sec.\,\ref{subsec:SES_classifications} that all such SPT phases must have order two, which is indeed the case because two copies of the hourglass-fermion phase will have a $\ZZZ_4$ index $\chi = 2 + 2 \equiv 0 \mod 4$. On the other hand, both square roots of the hourglass-fermion phase have a 3D $\ZZZ_2$ index $\nu_0 = 1$, so while they represent nontrivial $\ZZZ \times G$-protected SPT phases, they are not weak with respect to glide. As $\ZZZ \times G$-protected SPT phases they do not have order 2 because the hourglass-fermion phase is nontrivial. They have order 4 because four copies of either square root has a $\ZZZ_4$ index $\chi = 4 \times 1$ or $4 \times 3 \equiv 0 \mod 4$. We see that the hourglass-fermion phase and its square roots fall into categories (ii) and (iv) of Corollary \ref{cor:quad-chotomy}, respectively.

Let us now include interaction-enabled \cite{Note51} fermionic SPT phases and demonstrate how, even though the classification without glide is nontrivial in both 2 and 3D, we can still deduce the complete 3D classification with glide using Proposition \ref{prp:SES_classifications}. It has been proposed that the complete classifications of 2D \cite{2dChiralBosonicSPT_erratum} and 3D \cite{WangChong_3DSPTAII} $G$-protected fermionic SPT phases, for the $G$ specified at the beginning of this subsection, are
\begin{eqnarray}
\SPT^2_f\paren{G} &\isomorphic& \ZZZ_2, \label{2DAII}\\
\SPT^3_f\paren{G} &\isomorphic& \ZZZ_2{\oplus}\ZZZ_2{\oplus}\ZZZ_2,
\end{eqnarray}
respectively, where the $\ZZZ_2$ in 2D is generated by the QSH phase, and the three $\ZZZ_2$'s in 3D are generated by a band insulator and two bosonic SPT phases, respectively. Inserting these into the short exact sequence in Proposition \ref{prp:SES_classifications}, we are led to the abelian group extension problem
\e{0 \rightarrow \ZZZ_2 \rightarrow\; ? \rightarrow \ZZZ_2\oplus\ZZZ_2\oplus\ZZZ_2 \rightarrow 0.\label{groupextension_AII}
}
The solution to this problem is not unique, as is evident from $H^2_{\rm sym}\paren{\ZZZ_2 \oplus \ZZZ_2 \oplus \ZZZ_2; \ZZZ_2} \isomorphic \ZZZ_2 \oplus \ZZZ_2 \oplus \ZZZ_2$. However, we know that the hourglass-fermion phase and its square roots are robust to interactions. We claim that, with this additional piece of information, a unique solution can be found.

Indeed, Corollary \ref{cor:direct_sum_decomposition} says that the unknown term must be a direct sum of $\ZZZ_4$'s and/or $\ZZZ_2$'s. We now show that there is exactly one $\ZZZ_4$ and two $\ZZZ_2$'s. By the remarks in Corollary \ref{cor:direct_sum_decomposition}, the only way for $\ZZZ\times G$-protected SPT phases of type (iv) of Corollary \ref{cor:quad-chotomy} to arise is for there to be a $\ZZZ_4$ summand. Since both square roots of the hourglass-fermion phase are of type (iv), there must be at least one $\ZZZ_4$. On the other hand, each $\ZZZ_4$ contains an SPT phase that is weak with respect to glide, which must arise from an independent non-trivial SPT phase in one lower dimensions through the alternating-layer construction.
Since the classification in one lower dimensions is given by a single $\ZZZ_2$, there can be at most one $\ZZZ_4$ in the second term of Eq.\,(\ref{groupextension_AII}).
As a result, there is exactly one $\ZZZ_4$. This $\ZZZ_4$ maps onto one of the three $\ZZZ_2$'s in the third term of Eq.\,(\ref{groupextension_AII}). To make the map surjective as required by exactness, we need two additional $\ZZZ_2$'s in the second term of Eq.\,(\ref{groupextension_AII}), whose nontrivial elements are of type (iii) of Corollary \ref{cor:quad-chotomy}.

In conclusion, the \emph{complete} classification of 3D $\ZZZ \times G$-protected fermionic SPT phases, for the $G$ specified at the beginning of this subsection, is
\e{\SPT^3_f\paren{\ZZZ {\times} G} \isomorphic \ZZZ_4\oplus \ZZZ_2 \oplus \ZZZ_2,\label{solution_AII}}
where without loss of generality we can identify the nontrivial elements of $\ZZZ_4$ with the hourglass-fermion phase and its square roots. This represents one key result of this work, which goes beyond the known classification of the \emph{subset} of SPT phases that are weak with respect to glide \cite{Lu_sgSPT}:
\begin{equation}
\wSPT^3_f\paren{\ZZZ {\times} G} \isomorphic \ZZZ_2.
\end{equation}
We may anyway verify that this weak classification, together with Eq.\,(\ref{2DAII}), is consistent with Corollary \ref{cor:classification_SPT_weak_wrt_glide}. We remark that while we used such physical terms as ``weak with respect to glide" in our argument above, we could have derived Eq.\,(\ref{solution_AII}) purely mathematically, by combining an explicit classification of abelian group extensions of $\ZZZ_2 \oplus \ZZZ_2 \oplus \ZZZ_2$ by $\ZZZ_2$ with the requirement that the extension contain an element of order 4.

\section{Applications to bosonic SPT phases\label{sec:computations}}

\begin{table*}
\caption{Classification of bosonic SPT phases with glide reflection or translational symmetry. $\SPT^d_b\paren{G,\phi}$ is computed from the proposal (\ref{bSPT_0})-(\ref{bSPT_3}), whence the next three rows are deduced using Eqs.\,(\ref{GSPT}), (\ref{SES}), and (\ref{strong_weak}), respectively. Abelian group extensions are in general not unique, accounting for the non-uniqueness of some entries, which we parenthesize. ``Glide" and ``transl." indicate whether $\ZZZ$ is generated by a glide or a translation. $\SPT^4_b\paren{\ZZZ \times G, \phi}$ is left blank for both glide and translation because it requires $\SPT^4_b\paren{G, \phi}$ as an input, which we did not provide. The superscript $T$ in $\ZZZ_2^T$ indicates time reversal.
In the last column, we give physical models corresponding to the generators of underlined summands, where ``$E_8$" stands for the $E_8$ model \cite{Kitaev_honeycomb, 2dChiralBosonicSPT, 2dChiralBosonicSPT_erratum, Kitaev_KITP}, ``BIQH" for bosonic integer quantum Hall \cite{3dBTScVishwanathSenthil, BIQH}, ``3D $E_8$" for the 3D $E_8$ model \cite{3dBTScVishwanathSenthil, 3dBTScWangSenthil, 3dBTScBurnell}, and ``Haldane" for the Haldane chain \cite{AKLT, PhysRevLett.50.1153, Haldane_NLSM, Affleck_Haldane, Haldane_gap}.%
}
\label{table:bSPT}
\begin{tabular}{|c|ccccc|c|}
\hline
\multirow{2}{*}{Bosonic, $(G,\phi)=0$} & \multicolumn{5}{c|}{Spatial dimension $d$} & \multirow{2}{*}{Comments}\\
\cline{2-6}
 & 0 & 1 & 2 & 3 & 4 & \\
\hline
$\SPT^d_b(G,\phi)$ & 0 & 0 & $\underline{\ZZZ}$ & 0 & & $E_8$  \\
$\wSPT^d_b(\ZZZ\times G,\phi)$ & 0 & 0 & 0 & $\ZZZ_2$ & 0 & \\
$\SPT^d_b(\ZZZ\times G, \phi)$, glide & 0 & 0 & 0 & $\ZZZ_2$ & & \\
$\SPT^d_b(\ZZZ\times G, \phi)$, transl. & 0 & 0 & $\ZZZ$ & $\ZZZ$ & & \\
\hline
\hline
\multirow{2}{*}{Bosonic, $(G,\phi)=U(1)$} & \multicolumn{5}{c|}{Spatial dimension $d$} & \multirow{2}{*}{Comments} \\
\cline{2-6}
 & 0 & 1 & 2 & 3 & 4 & \\
\hline
$\SPT^d_b(G,\phi)$ & $\ZZZ$ & 0 & $\underline{\ZZZ} \oplus \underline\ZZZ$ & 0 & & BIQH, $E_8$ \\
$\wSPT^d_b(\ZZZ\times G,\phi)$ & 0 & $\ZZZ_2$ & 0 & $\ZZZ_2\oplus \ZZZ_2$ & 0 & \\
$\SPT^d_b(\ZZZ\times G,\phi)$, glide & 0 & $\ZZZ_2$ & 0 & $\ZZZ_2\oplus \ZZZ_2$ & & \\
$\SPT^d_b(\ZZZ\times G,\phi)$, transl. & $\ZZZ$ & $\ZZZ$ & $\ZZZ\oplus\ZZZ$ & $\ZZZ\oplus \ZZZ$ & & \\
\hline
\hline
\multirow{2}{*}{Bosonic, $(G,\phi)=\ZZZ_2^T$} & \multicolumn{5}{c|}{Spatial dimension $d$} & \multirow{2}{*}{Comments} \\
\cline{2-6}
 & 0 & 1 & 2 & 3 & 4 & \\
\hline
$\SPT^d_b(G,\phi)$ & 0 & $\ZZZ_2$ & 0 & $\ZZZ_2 \oplus \underline{\ZZZ_2}$ & & 3D $E_8$ \\
$\wSPT^d_b(\ZZZ\times G,\phi)$ & 0 & 0 & $\ZZZ_2$ & 0 & $\ZZZ_2 \oplus \ZZZ_2$ & \\
$\SPT^d_b(\ZZZ\times G,\phi)$, glide & 0 & $\ZZZ_2$ & $\ZZZ_2$ & $\ZZZ_2\oplus \ZZZ_2$ & & \\
$\SPT^d_b(\ZZZ\times G,\phi)$, transl. & 0 & $\ZZZ_2$ & $\ZZZ_2$ & $\ZZZ_2\oplus \ZZZ_2$ & & \\
\hline
\hline
\multirow{2}{*}{Bosonic, $(G,\phi)=\ZZZ_{N<\infty}$} & \multicolumn{5}{c|}{Spatial dimension $d$} & \multirow{2}{*}{Comments} \\
\cline{2-6}
 & 0 & 1 & 2 & 3 & 4 & \\
\hline
$\SPT^d_b(G,\phi)$ & $\ZZZ_N$ & 0 & $\ZZZ_N\oplus \underline{\ZZZ}$ & 0 & & $E_8$\\
$\wSPT^d_b(\ZZZ\times G,\phi)$ & 0 & $\ZZZ_{\gcd(N,2)}$ & 0 & $\ZZZ_{\gcd(N,2)}\oplus \ZZZ_2$ & 0 & \\
$\SPT^d_b(\ZZZ\times G,\phi)$, glide & $\ZZZ_{\gcd(N,2)}$ & $\ZZZ_{\gcd(N,2)}$ & $\ZZZ_{\gcd(N,2)}\oplus \ZZZ_2$ & $\ZZZ_{\gcd(N,2)} \oplus \ZZZ_2$ & & \\
$\SPT^d_b(\ZZZ\times G,\phi)$, transl. & $\ZZZ_N$ & $\ZZZ_N$ & $\ZZZ_N\oplus \ZZZ$ & $\ZZZ_N \oplus \ZZZ$ & & \\
\hline
\hline
\multirow{2}{*}{Bosonic, $(G,\phi)=\ZZZ_2 \times \ZZZ_2$} & \multicolumn{5}{c|}{Spatial dimension $d$} & \multirow{2}{*}{Comments} \\
\cline{2-6}
 & 0 & 1 & 2 & 3 & 4 & \\
\hline
$\SPT^d_b(G,\phi)$ & $\ZZZ_2^2$ & $\underline{\ZZZ_2}$ & $\ZZZ_2^3\oplus \underline\ZZZ$ & $\ZZZ_2^2$ & & Haldane, $E_8$\\
$\wSPT^d_b(\ZZZ\times G,\phi)$ & 0 & $\ZZZ_2^2$ & $\ZZZ_2$ & $\ZZZ_2^3 \oplus \ZZZ_2$ & $\ZZZ_2^2$ & \\
$\SPT^d_b(\ZZZ\times G,\phi)$, glide & $\ZZZ_2^2$ & ($\ZZZ_2 \oplus \ZZZ_4$ or $\ZZZ_2^3$) & ~($\ZZZ_2^2 \oplus \ZZZ_4$ or $\ZZZ_2^4$)~ & ~($\ZZZ_2^2 \oplus \ZZZ_4^2$ or $\ZZZ_2^4 \oplus \ZZZ_4$ or $\ZZZ_2^6$)~ & & \\
$\SPT^d_b(\ZZZ\times G,\phi)$, transl. & $\ZZZ_2^2$ & $\ZZZ_2^2\oplus\ZZZ_2$ & $\ZZZ_2\oplus\ZZZ_2^3\oplus \ZZZ$ & $\ZZZ_2^3\oplus\ZZZ\oplus \ZZZ_2^2$ & & \\
\hline
\hline
\multirow{2}{*}{Bosonic, $(G,\phi)=SO(3)$} & \multicolumn{5}{c|}{Spatial dimension $d$} & \multirow{2}{*}{Comments} \\
\cline{2-6}
 & 0 & 1 & 2 & 3 & 4 & \\
\hline
$\SPT^d_b(G,\phi)$ & 0 & $\underline{\ZZZ_2}$ & $\ZZZ\oplus \underline\ZZZ$ & 0 & & Haldane, $E_8$\\
$\wSPT^d_b(\ZZZ\times G,\phi)$ & 0 & 0 & $\ZZZ_2$ & $\ZZZ_2\oplus \ZZZ_2$ & 0 & \\
$\SPT^d_b(\ZZZ\times G,\phi)$, glide & 0 & $\ZZZ_2$ & $\ZZZ_2$ & $\ZZZ_2\oplus \ZZZ_2$ & & \\
$\SPT^d_b(\ZZZ\times G,\phi)$, transl. & 0 & $\ZZZ_2$ & $\ZZZ_2\oplus \ZZZ \oplus \ZZZ$ & $\ZZZ \oplus \ZZZ$ & & \\
\hline
\end{tabular}
\end{table*}

In \s{sec:applications}, we exemplified how one can utilize Proposition \ref{prp:SES_classifications} to deduce the classification of $\paren{\ZZZ\times G}$-protected fermionic SPT phases (with $\ZZZ$ generated by glide) from proposed classifications of $G$-protected fermionic SPT phases in the literature. The problem of identifying the correct classification was reduced an abelian group extension problem, which required very little technical work in comparison to deriving the classification from scratch. In this section, we apply the same principle to bosonic SPT phases for a variety of symmetries.

The input, $\SPT^d_b(G,\phi)$, of our computations will be given by a generalized cohomology theory $h_b$ that, in low dimensions, reads
\begin{eqnarray}
h^0_b\paren{BG,\phi} &=& H^2\paren{BG; \ZZZ,\phi} \label{bSPT_0}, \\
h^1_b\paren{BG,\phi} &=& H^3\paren{BG; \ZZZ,\phi} \label{bSPT_1}, \\
h^2_b\paren{BG,\phi} &=& H^4\paren{BG; \ZZZ,\phi} \oplus H^0\paren{BG; \ZZZ,\phi}, \label{bSPT_2} \\
h^3_b\paren{BG,\phi} &=& H^5\paren{BG; \ZZZ,\phi} \oplus H^1\paren{BG; \ZZZ,\phi}. \label{bSPT_3}
\end{eqnarray}
These expressions can be derived using a scheme due to Kitaev from a presumed classification of SPT phases \emph{without} symmetry. More specifically, we assume that bosonic $G$-protected SPT phases for trivial $G$ are classified by
\begin{eqnarray}
\SPT^{0,1,2,3}_b\paren{0} \isomorphic 0, ~0, ~\ZZZ, ~0,\label{bSPT_input}
\end{eqnarray}
in 0, 1, 2, and 3 dimensions, respectively, where the $\ZZZ$ in 2 dimensions is generated by the $E_8$ phase \cite{Kitaev_honeycomb, 2dChiralBosonicSPT, 2dChiralBosonicSPT_erratum, Kitaev_KITP}; this is consistent with the proposal reviewed in Ref.\,\cite{Wen_review_2016}, which goes up to 6 dimensions. As pointed out by Kitaev \cite{Kitaev_Stony_Brook_2011_SRE_1, Kitaev_Stony_Brook_2013_SRE, Kitaev_IPAM}, from the classification without symmetry one can reconstruct a not necessarily unique generalized cohomology theory $h$ which in turn will give one the classification for \emph{arbitrary} symmetries. In the case of Eq.\,(\ref{bSPT_input}), the reconstruction turns out to be unique in low dimensions, giving Eqs.\,(\ref{bSPT_0})-(\ref{bSPT_3}) \cite{Xiong}.

The output of our computations will be $\wSPT^d_b(\ZZZ \times G, \phi)$ for $d\leq 4$ and $\SPT^d_b(\ZZZ \times G, \phi)$ for $d\leq 3$, where $\ZZZ$ is generated by glide. These will be computed from $\SPT^d_b(G,\phi)$ using the correspondence (\ref{GSPT}) and the short exact sequence (\ref{SES}), respectively. We have summarized the results in Table \ref{table:bSPT}. As we can see, in most cases the short exact sequence (\ref{SES}) determines the classification of $d$-dimensional $\paren{\ZZZ\times G}$-protected bosonic SPT phases completely. The results for $\wSPT^3_b(\ZZZ \times G,\phi)$ are in agreement with Ref.\,\cite{Lu_sgSPT}.

\section{Discussions\label{sec:discussions}}

\subsection{Spatiotemporal glide symmetry}\la{sec:temporalglide}

In this paper we have focused on spatial glide symmetry, but with the right definitions we expect the Twisted Generalized Cohomology Hypothesis (hence also Proposition \ref{prp:SES_classifications}) to also work for generalized, spatiotemporal glide symmetries, as long as they commute with the symmetry $G$.
An example of spatiotemporal symmetries would be a translation followed by a time reversal, which has been considered by the authors of Ref.\,\cite{Teo_AF_TRS} under the name ``antiferromagnetic time-reversal symmetry" (AFTRS). 2D and 3D topological superconductors in Atland-Zirnbauer class D are classified by $\ZZZ$ and $0$, respectively, where the $\ZZZ$ in 2D is generated by spinless $p+ip$ superconductors. By putting a spinless $p+ip$ superconductor on all planes of constant $x\in \ZZZ$ and its \emph{time-reversed} version (time reversal squares to the identity in this case due to spinlessness) on all $x\in \ZZZ + 1/2$ planes without coupling, one creates a 3D system that respects the fermionc-parity $\ZZZ_2^f$ and AFTRS. Since spinless $p+ip$ superconductors are robust to interactions and admit no square root, our Proposition \ref{prp:SES_classifications} implies that this 3D system must represent a nontrivial $\text{AFTRS}\times\ZZZ_2^f$-protected fermionic SPT phase. Indeed, this was argued in Ref.\,\cite{Teo_AF_TRS} to be the case through the construction of surface topological orders and subsequent confinement of the classical extrinsic defects among the anyons. Since it is believed that {spinless} $p+ip$ superconductors generate the \emph{complete} classification of 2D fermionic SPT phases {with only fermion-parity symmetry}, this gives a putative $\ZZZ_2$ classification of 3D fermionic {$\text{AFTRS}\times\ZZZ_2^f$-protected} SPT phases that are weak with respect to glide. The proposal that 3D fermionic SPT phases {with only fermion-parity symmetry} have a trivial classification \cite{Kapustin_Fermion} would further imply that this $\ZZZ_2$ actually classifies \emph{all} 3D {$\text{AFTRS}\times\ZZZ_2^f$-protected} fermionic SPT phases.

The case of a temporal translation followed by a spatial reflection has been studied in the context of non-interacting Floquet topological phases \cite{time_glide}, but we shall leave this to future works in view of the subtleties in the definition of a Floquet phase.

\subsection{Pure translation versus glide reflection}\la{sec:puretranslation}

A glide reflection is a translation followed by a reflection. In this section, we set out to answer two questions: (a) how has this additional reflection complicated the classification of SPT phases? The symmetry group generated by glide reflection contains a subgroup of pure translations. (b) What would happen if we relaxed glide symmetry to its translational subgroup? 

The first question, (a), can be answered by contrasting Proposition \ref{prp:SES_classifications} with the analogous result for pure translations \cite{Xiong}:
\begin{prp}
\label{prp:strong_weak}
Assume the Twisted Generalized Cohomology Hypothesis. Let $\ZZZ$ be generated by a \emph{translation} and $G$ be arbitrary 
\footnote{In Ref.\,\cite{Xiong}, $G$ was assumed to preserve spacetime orientation, but the proof of Proposition \ref{prp:SES_classifications} in App.\,\ref{app:proof} can be easily adapted to show that this restriction was unnecessary.}.
There is a \emph{split} short exact sequence,
\begin{equation}
0 \fromto \SPT^{d-1}(G) \fromto \SPT^d\paren{\ZZZ\times G} \fromto  \SPT^d\paren{G} \fromto 0.\label{SES_translation}
\end{equation}
In particular, there is an isomorphism,
\begin{equation}
 \SPT^d\paren{\ZZZ\times G} \isomorphic \SPT^{d-1}(G)  \oplus \SPT^d\paren{G}.\label{strong_weak}
\end{equation}
\end{prp}

\noindent In sequence (\ref{SES_translation}), the first map is given by a layer construction, which is the same as the alternating-layer construction but without the orientation-reversal that occurs every other layer. The second map is given by forgetting translational symmetry. Unlike sequence (\ref{SES}), which has factors of 2 in the first and third terms, the sequence for pure translation does not contain any factors of 2. Starting from $\ZZZ\times G$-protected SPT phases, forgetting the translational symmetry gives us all $G$-protected SPT phases, so all $G$-protected SPT phases are compatible with translational symmetry, even if they do not square to the trivial phase.
Starting from a nontrivial $G$-protected SPT phase in one lower dimensions, applying the layer construction will always give us a nontrivial $\ZZZ\times G$-protected SPT phase, even if the lower-dimensional phase admits a square root. Furthermore, sequence (\ref{SES_translation}) is split. This means its second term is completely determined by the first and third terms, according to Eq.\,(\ref{strong_weak}), and there is no abelian group extension problem to solve. The orientation-reversing nature of glide reflections is responsible for all the complications in sequence (\ref{SES}). Various examples of $\SPT^d(\ZZZ \times G)$, with $\ZZZ$ generated by glide or pure translation, are juxtaposed in Tables \ref{table:bSPT}. 

Regarding the second question, (b), denoting the glide symmetry by $\ZZZ$, we know that if a $\ZZZ\times G$-protected SPT phase becomes trivial when $\ZZZ$ is relaxed to its translational subgroup, then it must become trivial when $\ZZZ$ is forgotten altogether. It turns out that the converse is also true:

\begin{prp}
Assume the Twisted Generalized Cohomology Hypothesis. Let $\ZZZ$ be generated by a glide reflection and $G$ be arbitrary. If a $\ZZZ\times G$-protected SPT phase becomes trivial under glide forgetting, then it must already become trivial when $\ZZZ$ is relaxed to its translational subgroup.
\label{prp:glide_to_translation}
\end{prp}

\noindent This can be either proved mathematically from the Twisted Generalized Cohomology Hypothesis as in App.\,\ref{app:proof_translation}, or argued physically as we proceed to do.  Let us denote the translational subgroup of the glide symmetry $\ZZZ$ by $2\ZZZ$ and introduce the map,
\begin{equation}
\gamma: \SPT^d\paren{\ZZZ \times G} \fromto \SPT^d\paren{2\ZZZ \times G},\label{gamma}
\end{equation}
given by relaxing glide to its translational subgroup. By the same argument as in Fig.\,\ref{fig:image_alpha_subset_kernel_beta}, a system obtained through the alternating-layer construction can always be made trivial while preserving the translational symmetry. This means $\image \alpha \subset \kernel \gamma$.  On the other hand, the fact that a $\ZZZ\times G$-protected SPT phase that becomes trivial when $\ZZZ$ is relaxed to $2\ZZZ$ must become trivial when $\ZZZ$ is forgotten altogether shows that $\kernel \gamma \subset \kernel \beta$. Since we know that a $\ZZZ\times G$-protected SPT phase is weak with respect to glide if and only if it can be obtained through the alternating-layer construction, that is, $\image \alpha= \kernel \beta$, we must have $\image \alpha = \kernel \gamma = \kernel \beta$, whence the desired result follows.

\subsection{{Time reversal and inverse of SPT phase in zero dimension}}

{In Sec.\,\ref{sec:temporalglide}, we remarked that we expect our results to hold for spatiotemporal glide symmetries such as spatial translation followed by time reversal. For the same physical intuition in Sec.\,\ref{sec:physical_picture} to apply, it is crucial for time reversal to give the inverse of an SPT phase protected by on-site symmetry. While in positive dimensions one needs to specify how time reversal acts on various degrees of freedom, for a 0-dimensional SPT phase the action of time reversal is unique (up to a phase). The latter makes an explicit proof that time reversal gives the inverse of an SPT phase possible, which we present below. Time reversal is in fact the only way to define inverses in 0 dimension, where there is no spatial coordinate to reverse.}

{To begin, we} note that a 0-dimensional SPT phase is nothing but an isomorphism class of \emph{1-dimensional} representations $\rho$ of $G$, owing to the uniqueness of ground state. These representations have the form
\begin{equation}
\rho_g = u_g K^{s(g)},
\end{equation}
where {$u_g\in U(1)$ is a complex number of unit modulus}, $K$ is complex conjugation, and $s: G \fromto \braces{0,1}$ is a homomorphism [same as $\phi$ but written additively: $\phi(g) = (-1)^{s(g)}$]. Stacking corresponds to taking the tensor product of two representations:
\begin{equation}
\paren{u_g K^{s(g)}} \otimes \paren{v_g K^{s(g)}} = u_g v_g K^{s(g)},
\end{equation}
{The trivial phase, i.e. the identity element under stacking, is obviously represented by
\begin{equation}
\identity_g \coloneq K^{s(g)}
\end{equation}
Being antiunitary, time reversal must act like
\begin{equation}
w K
\end{equation}
on the 1-dimensional subspace spanned by the ground state, for some $w\in U(1)$.}
Now, to prove that
\begin{equation}
{(wK) \rho_g (wK)^{-1}}
\end{equation}
represents the inverse SPT phase of what $\rho_g$ represents, we simply compute
{
\begin{eqnarray}
\rho_g \otimes \paren{(wK)\rho_g (wK)^{-1}} &=& \paren{u_g K^{s(g)}} \otimes \paren{w u_g^* K^{s(g)} w^{-1}} \nonumber\\
&=& w \identity_g w^{-1},
\end{eqnarray}
which is equivalent to the representation $\identity_g$ and hence represents the trivial phase.}

\subsection{What if we knew the robustness of only the hourglass-fermion phase} \la{sec:whatif}

In Sec.\,\ref{sec:hourglass_fermions}, we combined a corollary of Proposition \ref{prp:SES_classifications} with the fact that the QSH phase is robust to interactions and admits no interacting square root to show that the hourglass-fermion phase is also robust to interactions. Here we ask the converse question. That is, if we only knew that the hourglass-fermion phase is robust to interactions, could we deduce that the QSH phase is robust to interactions and admits no interacting square root?

This is important because an independent argument for the robustness of hourglass-fermion phase has been recently offered in Ref.\,\cite{Lu_sgSPT}. An affirmative answer to the above question would not only lend credence to the consistency of our minimalist framework, but also further corroborate existing arguments \cite{qi2008, essin2009, qi_spincharge,lee2008,aa2011, fu2006, Lu_sgSPT} for the robustness of various phases. We now show that this is indeed the case.

Assume the hourglass-fermion phase is robust to interactions, and suppose to the contrary that the QSH phase was not robust to interactions. Then there would exist a way to trivialize the QSH phase by turning on interactions while preserving the many-body gap and the $U(1)$ and $\T$ symmetries. Since the hourglass-fermion phase can be obtained from the QSH phase through the alternating-layer construction, one could then trivialize it as well by turning on \emph{intra}-layer interactions for all layers at once, while preserving the many-body gap, the glide symmetry, as well as $U(1)$ and $\T$. This would contradict the assumption that the hourglass-fermion phase is robust to interactions.

There is also a more formal way of looking at this. To do so, we need to first recognize that the passage from the classification of {translation-invariant} topological insulators or superconductors to the classification of SPT phases defines a homomorphism between abelian groups. For instance, we have homomorphisms
\begin{eqnarray}
i_1'&:& \ZZZ_2 \fromto A, \\
i_2&:& \ZZZ_4 \fromto B, \la{definei2}\\
i_3'&:& \ZZZ_2 \fromto C,
\end{eqnarray}
where the three domains are respectively the strong classification of 2D translation-invariant TIs in class AII, the strong classification of 3D translation-invariant, glide-symmetric TIs in class AII, and the strong classification of 3D translation-invariant TIs in class AII without the glide constraint. The three codomains are respectively
\begin{eqnarray}
A &\coloneq& \SPT^2\paren{G},\\
B &\coloneq& \SPT^3\paren{\ZZZ \times G},\\
C &\coloneq& \SPT^3\paren{G}, 
\end{eqnarray}
where $\ZZZ$ is generated by glide and $G$ is generated by $U(1)$ and $\T$ where $\T$ squares to fermion parity.
These homomorphisms induce homomorphisms between quotients and torsions, so now we have homomorphisms
\begin{eqnarray}
i_1&:& \ZZZ_2 = \ZZZ_2 / 2 \ZZZ_2 \fromto A/2A, \\
i_3&:& \ZZZ_2 = \braces{\nu_0 \in \ZZZ_2 | 2\nu_0 = 0} \fromto \braces{[c] \in C | 2[c] = 0}.
\end{eqnarray}
Homomorphisms $i_1$, $i_2$, and $i_3$ fit into a commutative diagram,
\begin{eqnarray}
\adjustbox{valign=M}{\includegraphics{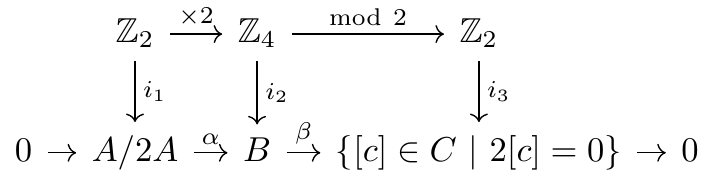}},\la{commutativediag}
\end{eqnarray}
where the first row is \q{SES_AII}, and the indeterminate second row is exact by Proposition \ref{prp:SES_classifications}. The argument in Sec.\,\ref{subsec:robustness_hourglass_fermion_phase} amounts to the implication
\begin{equation}
i_1 \mbox{ injective} \Rightarrow i_2 \mbox{ injective},
\end{equation}
whereas the argument at the beginning of this subsection amounts to the implication
\begin{equation}
i_2 \mbox{ injective} \Rightarrow i_1 \mbox{ injective}.
\end{equation}
These can be easily proved through diagram chasing. It is also an easy exercise to show that
\begin{equation}
i_1 \mbox{ or } i_2 \mbox{ injective} \Rightarrow i_3 \mbox{ injective}.
\end{equation}

\section{Summary and outlook} \la{sec:summaryoutlook}

In summary, we have derived a short exact sequence [\q{SES}] of three abelian groups that in consecutive order describe (i) the classification of $(d-1)$-dimensional fermionic/bosonic SPT phases with an arbitrary symmetry group $G$ in an arbitrary dimension $d$, (ii) the classification of $d$-dimensional SPT phases with symmetry $\ZZZ \times G$, where $\ZZZ$ is generated by a glide, and (iii) the classification $d$-dimensional SPT phases with symmetry $G$. An alternating-layer construction (see Fig.\,\ref{fig:alternating_layer_construction}) maps (i) into (ii), which in turn is mapped under glide-forgetting to (iii). We considered both spatial glide transformations and spatiotemporal ones (Sec.\,\ref{sec:temporalglide}), such as a spatial translation followed by time reversal.

We demonstrated how the structure of the short exact sequence constrains the classification of SPT phases. First off, we showed that for any given $G$, classification (ii) can only be a direct sum of $\ZZZ_4$ and $\ZZZ_2$ subgroups (see Corollaries \ref{cor:quad-chotomy} and \ref{cor:direct_sum_decomposition}). We further recognized, given (i) and (iii) as input, that the task of determining (ii) using the short exact sequence simply constitutes an abelian group extension problem, which requires little work compared to determining (ii) using other methods, e.g., formulating topological invariants. In some cases, the abelian group extension problem has a unique solution and determines (ii) completely, as we exemplified using fermionic SPT phases in Wigner-Dyson class A for spatial glide (see \s{subsec:classA}) and Altland-Zirnbauer class D for spatiotemporal glide (see \s{sec:temporalglide}), and bosonic SPT phases for a variety of symmetries (see \s{sec:computations}). In other cases, it has more than one admissible solutions. We elucidated the latter phenomenon using fermionic SPT phases in the Wigner-Dyson class AII (see Sec.\,\ref{subsec:sanity_check}), and bosonic SPT phases with $\ZZZ_2 \times \ZZZ_2$ symmetry (see Sec.\,\ref{sec:computations}). While the solution is not unique in these cases, a little additional input can often fully remove the ambiguity. Supplemented with the robustness of certain \textit{free-fermion} phases to interactions (see Sec.\,\ref{sec:hourglass_fermions}), we selected out of the many the one solution that provides the complete classification of 3D glide-symmetric SPT phases in Wigner-Dyson class AII: $\ZZZ_4 \oplus \ZZZ_2 \oplus \ZZZ_2$, as described in Sec.\,\ref{sec:computations}. In our reasoning, we utilized the relationship between free-fermion classifications and (interacting) SPT classifications, which we encapsulated in the commutative diagram (\ref{commutativediag}). This diagram delineated a ``map" from a possibly non-exact sequence of non-interacting classifications to our short exact sequence of interacting classifications, and encoded the robustness of various phases in the injectivity of certain homomorphisms. The latter made it transparent that the robustness of the hourglass fermion phase to interactions is closely related to the nonexistence of square root of the QSH phase and to the robustness of the 3D strong TI in the same class. In fact, the former two imply each other, and either of them implies the third (see \s{sec:whatif}). We pointed out that known arguments for the robustness of the QSH phase can be adapted to show that it has no square root, and from that we deduced that the hourglass-fermion phase must be robust to interactions (Sec.\,\ref{sec:hourglass_fermions}).

The short exact sequence afforded us further implications. For one, a $G$-protected SPT phase in (iii) has a glide-symmetric representative if and only if it squares to the trivial phase. For another, a $\ZZZ \times G$-protected SPT phases becomes trivial under glide-forgetting if and only if it can be obtained through the alternating-layer construction; we have referred to such phases as phases that are weak with respect to glide and clarified in which cases (e.g., class A with spatial glide or class D with spatiotemporal glide) the ``weak classification" is complete and in which cases it is not (e.g., class AII with spatial glide). Furthermore, we proved that if a $\ZZZ \times G$-protected SPT phase becomes trivial when glide is forgotten altogether, then it must become trivial as soon as glide is relaxed to its translational subgroup (see \s{sec:puretranslation}). We also contrasted our short exact sequence for glide, which may or may not be split, to an analogous short exact sequence for pure translation, which is always split (see \s{sec:puretranslation}).

We are hopeful that the machinery of our minimalism framework would spawn the proof of more nontrivial results concerning the classification of SPT phases. Hints for the potential existence of more relations like our short exact sequence can be found in Refs.\,\cite{Hermele_torsor,Huang_dimensional_reduction}. The authors of these papers described a number of ways to build (or reduce) higher-dimensional spatial SPT phases from (resp.\,to) lower-dimensional SPT phases. For instance, one can build a 3D system that respects the reflection symmetry $(x,y,z) \mapsto (-x,y,z)$ by putting a 2D system with on-site $\ZZZ_2$ symmetry on the $x=0$ plane and sandwiching it by trivial half-infinite systems on both sides. Conversely, one can reduce certain 3D reflection-symmetric system to a 2D system with on-site $\ZZZ_2$ symmetry living on the $x=0$ plane by trivializing the $x\neq 0$ regions symmetrically. These physical considerations suggest there may be a general relation among the classification of $d$-dimensional $G$-protected SPT phases, $d$-dimensional $\ZZZ_2^P \times G$-protected SPT phases, and $(d-1)$-dimensional $\ZZZ_2 \times G$-protected SPT phases, where $\ZZZ_2^P$ is represented antiunitarily and $\ZZZ_2$ is represented unitarily. We hope the minimalist framework will help us put our finger on the precise form of these new relations.

{\emph{Note added.}---While this manuscript was in press, the preprint \cite{Gaiotto_Johnson-Freyd} appeared, which demonstrated how the structure of a generalized cohomology theory could arise naturally in the classification of SPT phases by consideration of decorated defects of various codimensions. This provides further, strong evidence that the Twisted Generalized Cohomology Hypothesis is correct.}

\begin{acknowledgments}
We are grateful to Ashvin Vishwanath, Yuan-Ming Lu, Jeffrey Teo, and Hoi-Chun Po for inspiring discussions, {and the referees for their valuable comments}.
AA was supported by the Yale Postdoctoral Prize Fellowship.\\
\end{acknowledgments}

\section*{Organization of Appendix}

In App.\,\ref{app:twisted_generalized_cohomology}, we review generalized cohomology theories and formulate the Twisted Generalized Cohomology Hypothesis. In App.\,\ref{app:relationship_Z4_Z2}, we prove Eq.\,(\ref{glideforget}), which relates the $\ZZZ_4$ and $\ZZZ_2$ invariants in class AII. In App.\,\ref{app:proof}, we prove Proposition \ref{prp:SES_classifications}. In App.\,\ref{app:proof_corollaries}, we prove its Corollaries \ref{cor:quad-chotomy} and \ref{cor:direct_sum_decomposition}. In App.\,\ref{app:proof_translation}, we prove Proposition \ref{prp:glide_to_translation}, which concerns the relaxation of glide to its translational subgroup.

\appendix

\section{The Twisted Generalized Cohomology Hypothesis\label{app:twisted_generalized_cohomology}\label{app:TGCH}}

The minimalist framework is founded on the Twisted Generalized Cohomology Hypothesis, which is a twisted version of the Generalized Cohomology Hypothesis one of us formulated in Ref.\,\cite{Xiong}. It is based on Kitaev's argument that the classification of SPT phases should carry the structure of generalized cohomology theories \cite{Kitaev_Stony_Brook_2011_SRE_1, Kitaev_Stony_Brook_2013_SRE, Kitaev_IPAM}. This appendix will furnish us with necessary mathematical background with which we will then formulate the Hypothesis. A more thorough discussion on generalized cohomology theories can be found in Ref.\,\cite{Xiong} and many textbooks \cite{Hatcher, DavisKirk, Adams1, Adams2}.

A generalized cohomology theory $h$ can be represented by an $\Omega$-spectrum $F$, which by definition is a sequence
\begin{equation}
\ldots, F_{-2}, F_{-1}, F_0, F_1, F_2, \ldots
\end{equation}
of pointed topological spaces together with pointed homotopy equivalences
\begin{equation}
F_n \homotopic \Omega F_{n+1},\label{FOmegaF}
\end{equation}
where $\Omega$ is the loop space functor. In the non-twisted case, the generalized cohomology theory $h$ outputs an abelian group $h^n(X)$ for each given topological space $X$ and integer $n$, according to
\begin{equation}
h^n(X) \coloneq \brackets{X, \Omega F_{n+1}}.
\end{equation}
Here $[X, Y]$ denotes the set of homotopy classes of maps from $X$ to $Y$; when $Y$ comes from an $\Omega$-spectrum this set can be endowed with an abelian group structure.

In the twisted case, one is given an integer $n$, a pointed topological space $X$, and an action $\phi_X$ of the fundamental group $\pi_1(X)$ on the $\Omega$-spectrum $F$. The generalized cohomology theory $h$ then outputs an abelian group according to
\begin{equation}
h^n\paren{X, \phi_X} \coloneq \brackets{ \widetilde X, \Omega F_{n+1} }_{\pi_1(X)}. \label{widetilde_X_Omega_F_n+1}
\end{equation}
Here $\widetilde X$ denotes the universal cover of $X$, and $[X, Y]_G$ denotes the set of homotopy classes of $G$-equivariant maps from $X$ to $Y$. Again, when $Y$ comes from an $\Omega$-spectrum this set can be endowed with an abelian group structure. Recall that a $G$-equivariant map $f: X \fromto Y$ is a map that commutes with the action of $G$, i.e. $g.(f(x)) = f(g.x)~\forall g\in G$ and $x \in X$. It is a simple exercise to show that if $\pi_1(X)$ acts trivially on the $\Omega$-spectrum, then $h^n(X, \phi_X) = h^n(X)$.

Each $h^n$ is functorial, or more specifically, contravariant, which means the following. In the non-twisted case, maps between topological spaces,
\begin{equation}
f: X \fromto Y,
\end{equation}
induce homomorphisms between abelian groups,
\begin{equation}
f^*: h^n(Y) \fromto h^n(X), \label{temp010}
\end{equation}
such that the coherence relations $\paren{f_2 \circ f_1}^* = f_1^* \circ f_2^*$ and $\identity^*=\identity$ are satisfied. The induced homomorphism $f^*$ is given by precomposing maps $Y \fromto \Omega F_{n+1}$ with $f$. In the twisted case, one requires $f$ to additionally respect the fundamental group actions. That is, if $f_*: \pi_1(X) \fromto \pi_2(Y)$ is the homomorphism induced by $f$ and $\phi_X: \pi_1(X) \fromto \Aut(F)$ and $\phi_Y: \pi_1(Y) \fromto \Aut(F)$ are the fundamental group actions of $X$ and $Y$ on $F$, then one requires
\begin{equation}
\phi_Y \circ f_* = \phi_X.
\end{equation}
We denote $f$'s that satisfy this constraint by
\begin{equation}
f: (X, \phi_X) \fromto (Y, \phi_Y).
\end{equation}
For such $f$'s, the the same kind of precomposition gives rise to a homomorphism between abelian groups,
\begin{equation}
f^*: h^n(Y, \phi_Y) \fromto h^n(X, \phi_X),
\end{equation}
satisfying the same coherence relations.

For the purpose of classifying {\emph{bosonic}} SPT phases, we will set $X$ to be the the classifying space $BG$ of the symmetry group $G$, and $\phi_X$ according to how $G$ is represented. Let
\begin{equation}
\phi: G \fromto \{\pm 1\}
\end{equation}
be the homomorphism that sends antiunitarily represented elements to $-1$ and unitarily represented elements to 1. By continuity, $\phi$ can be viewed as a homomorphism that goes from $\pi_0(G)$ to $\braces{\pm 1}$ instead. It is an elementary property of classifying spaces that $\pi_1(BG) \isomorphic \pi_0(G)$. Therefore, $\phi$ can be viewed as a homomorphism that goes from
\begin{equation}
\phi: \pi_1(BG) \fromto \{\pm 1\}. \label{phi_G:pi_1(BG)}
\end{equation}
On the other hand, there is a canonical action of $\braces{\pm 1}$ on $F$, where the nontrivial element of $\braces{\pm 1}$ sends every loop $l\in \Omega F_{n+1}$ to the reverse loop $\bar l \in \Omega F_{n+1}$. Composing Eq.\,(\ref{phi_G:pi_1(BG)}) with this canonical action, we get the required action of $\pi_1(BG)$ on the $\Omega$-spectrum. We shall denote this action also by $\phi$ \footnote{In this physical setting, since all actions on $F$ factor through $\ZZZ_2$, one can use the $\ZZZ_2$-principal bundle {$\widetilde{BG}/ \kernel \phi \fromto BG$} instead of the universal cover {$\widetilde{BG} \fromto BG$} to obtain the same result. Alternatively, one can use the $G$-principal bundle $EG \fromto BG$.}.

For each non-negative integer $d$, we will then identify $h^d(BG, \phi)$ as the set of $d$-dimensional $G$-protected {bosonic} SPT phases, $\SPT^d_b(G,\phi)$. The abelian group structure of $h^d(BG, \phi)$ will be identified with the abelian group structure `+' of $\SPT^d_b(G,\phi)$ under stacking -- the binary operation that maps a pair of $d$-dimensional $G$-symmetric systems $a$ and $b$ with many-body Hilbert spaces $\mathscr H_{a,b}$, symmetry actions $\rho_{a,b}$ on the Hilbert spaces, Hamiltonians $\hat H_{a,b}$, and ground states $\ket{\Psi_{a,b}}$ to the $d$-dimensional $G$-symmetric system $a+b$ whose many-body Hilbert space, symmetry action, Hamiltonian, and ground state are
\begin{eqnarray}
\mathscr H_{a+b} &\coloneq& \mathscr H_a \otimes \mathscr H_b, \\
\rho_{a+b} &\coloneq& \rho_a \otimes \rho_b, \\
\hat H_{a+b} &\coloneq& \hat H_a \otimes \hat{\mathbb I}_b + \hat{\mathbb I}_a \otimes \hat H_b, \\
\ket{\Psi_{a+b}} &\coloneq& \ket{\Psi_a} \otimes \ket{\Psi_b},
\end{eqnarray}
respectively \cite{Note21}. The illustration in Fig.\,\ref{fig:stacking} is a mnemonic device, but the construction can be defined abstractly.

The functoriality of $h^d$, on the other hand, will be identified with a parallel property of the classification of {bosonic} SPT phases. Given any group homomorphism $f: (G_1,\phi_1) \fromto (G_2,\phi_2)$ satisfying $\phi_2 \circ f = \phi_1$ and a $G_2$-action $\rho_2: G_2 \fromto \Aut(\mathscr H)$ on Hilbert space $\mathscr H$, the composition $\rho_2 \circ f: G_1 \fromto \Aut(\mathscr H)$ defines a $G_1$-action on the same $\mathscr H$. This means $f$ can be used to convert $G_2$-symmetric systems to $G_1$-symmetric systems. Indeed, we can retain the same Hilbert space and Hamiltonian and simply replace $G_2$-actions by $G_1$-actions following the above recipe. In the event that $f: (G_1,\phi_1) \fromto (G_2,\phi_2)$ is an inclusion, this conversion process is precisely a symmetry-forgetting process. In general, $f$ does not have to be either injective or surjective, and the conversion process is a symmetry forgetting followed by a symmetry relabeling. Either way, we get an induced homomorphism,
\begin{equation}
f^*: \SPT^d_b(G_2, \phi_2) \fromto \SPT^d_b(G_1, \phi_1).
\end{equation}
This we will identify with the homomorphism that $f$ induces on $h^d$,
\begin{equation}
f^*: h^d(BG_2, \phi_2) \fromto h^d(BG_1, \phi_1). \label{temp020}
\end{equation}

{
For the purpose of classifying \emph{fermionic} SPT phases, further consideration must be given to the fermion-parity symmetry $\ZZZ_2^f$. The full symmetry group $G$ is an extension of $G/\ZZZ_2^f$ by $\ZZZ_2^f$. When the extension is split, i.e.\,$G = (G/\ZZZ_2^f) \times \ZZZ_2^f$, one can form a generalized cohomology theory that takes $G/\ZZZ_2^f$ as the input \cite{Wen_Fermion}; when the extension is non-split, there is a fermionic twisting that needs to go into the definition of generalized cohomology theory \cite{Kapustin_Fermion,Freed_SRE_iTQFT,fidkowski2011topological}. In order to avoid this complication, we use a formulation introduced in the recent paper Ref.\,\cite{Gaiotto_Johnson-Freyd}. Let us consider an arbitrary direct-product factorization of the full symmetry group,
\begin{eqnarray}
G &=& G_b \times G_f, \label{G=G_b_times_G_f}\\
\phi &=& \phi_b \times \phi_f, \label{phi=phi_b_times_phi_f}
\end{eqnarray}
such that $\ZZZ_2^f \subset G_f$. The idea of Ref.\,\cite{Gaiotto_Johnson-Freyd} is to treat $\paren{G_f,\phi_f}$ as a fixed parameter of the problem and let $\paren{G_b, \phi_b}$ vary. Since fermionic parity is in $G_f$ but not $G_b$, this puts the fact that there is fermionic matter in the theory effectively in a black box. With $\paren{G_f,\phi_f}$ held fixed, we are dealing exclusively with bosonic symmetries $\paren{G_b, \phi_b}$. Then we will identify $\SPT^d_f\paren{G,\phi}$ with $h^d\paren{BG_b,\phi_b}$, with the caveat that a different $h$ may need to be used for a different $\paren{G_f,\phi_f}$. For fixed $\paren{G_f,\phi_f}$, however, $h$ is a generalized cohomology theory. That is, $d$-dimensional $\paren{G_b\times G_f}$-protected fermionic SPT phases for fixed $\paren{G_f,\phi_f}$ but varying $d$ and $\paren{G_b, \phi_b}$ form a generalized cohomology theory.

The identification of abelian group structure and functoriality then proceeds almost identically to the bosonic case. We will identify the abelian group structure of $h^d\paren{BG_b,\phi_b}$ with the abelian group structure `+' of $\SPT^d_f\paren{G_b \times G_f,\phi_b \times \phi_f}$ under stacking. Given any $f: \paren{G_{b1}, \phi_{b1}} \fromto \paren{G_{b2}, \phi_{b2}}$, we will identify the homomorphism
\begin{equation}
f^*: h^d\paren{BG_{b2},\phi_{b2}} \fromto h^d\paren{BG_{b1},\phi_{b1}} \label{temp021}
\end{equation}
with the homomorphism
\begin{equation}
\SPT^d_f(G_{b2} \times G_f, \phi_{b2} \times \phi_f) \fromto \SPT^d_f(G_{b1} \times G_f, \phi_{b1} \times \phi_f)
\end{equation}
defined by trading $G_{b2}$-actions for $G_{b1}$-actions as before. Note that here we are not allowed for arbitrary homomorphisms $G_{b1} \times G_f \fromto G_{b2} \times G_f$ but only ones of the form $f \times \identity_{G_f}$, with $f: G_{b1} \fromto G_{b2}$ and $\identity_{G_f}$ the identity on $G_f$. This makes sense, as it is not only the group structure of $G$ and homomorphism $\phi$, but also the status of a central element of $G$ being fermion parity, that must be preserved.

We stress that, for a given $G$, how one separates $G$ into $G_b$ and $G_f$ is purely a matter of choice. When multiple factorizations exist, any of them can be used to compute $\SPT^d_f(G,\phi)$ and the result by hypothesis will be the same. When the $\ZZZ_2^f$-extension is split, i.e.\,$G = (G/\ZZZ_2^f) \times \ZZZ_2^f$ -- which is the setting of Ref.\,\cite{Wen_Fermion} -- the smallest possible $G_f$ is $\ZZZ_2^f$. When the extension is non-split, the smallest possible $G_f$ will be strictly larger than $\ZZZ_2^f$ \footnote{{Strictly speaking it is incorrect to say \emph{the} smallest possible $G_f$, as it may not be unique. For example, $\ZZZ_2 \times \ZZZ_4$, with $(0,2)$ the fermion parity, can be written as either $\braces{(0,0),(1,0)} \times \braces{(0,0),(1,1),(0,2),(1,3)}$ or $\braces{(0,0),(1,0)} \times \braces{(0,0),(0,1),(0,2),(0,3)}$.}}. This is the case for the symmetry group of hourglass fermions considered in this paper, which is generated $U(1)$, $\T$, and glide. We note that the fermion-parity symmetry $\ZZZ_2^f$ is contained in the $U(1)$ charge-conservation symmetry. Because $\T$ squares to fermion parity, i.e.\,the nontrivial element of $\ZZZ_2^f$, it must be included in $G_f$. Because the $q$th power of any element of the form $e^{i(2p+1)\pi/q}\in U(1)$ is the fermion parity, $e^{i(2p+1)\pi/q}$ must also be included in $G_f$. In fact, because $U(1)$ is a continuous group, the entire $U(1)$ must be included in $G_f$. On the other hand, the subgroup $\ZZZ$ generated by glide commutes with, and intersects trivially with, the subgroup generated by $U(1)$ and $\T$, so it can be kept in $G_b$. Note that regardless of what $G$ is, a factorization (\ref{G=G_b_times_G_f}) always exists: one can set $G_f=G$ and $G_b=0$ (the trivial group).}

The above are all that go into the Twisted Generalized Cohomology Hypothesis, which can be summed up as follows.

{
\begin{framednameddef}[Twisted Generalized Cohomology Hypothesis]
~\\
Bosonic.---There exists a generalized cohomology theory $h_b$ such that for each $d\in \NNN$ there are isomorphisms
\begin{equation}
\mathcal{SPT}^d_b\paren{G,\phi}\isomorphic h^d_b\paren{BG, \phi}\label{GCH_iso} 
\end{equation}
that are natural \footnote{{Here, naturality is in the sense of category theory. In the present case, it says the isomorphisms respect the functorial structure of the two sides.}} in $(G,\phi)$.\vspace{0.1in}

\noindent Fermionic.---For each $\paren{G_f, \phi_f}$ containing the fermion parity symmetry, there exists a generalized cohomology theory $h_{f, \paren{G_f,\phi_f}}$ such that for each $d\in \NNN$ there are isomorphisms
\begin{equation}
\mathcal{SPT}^d_f\paren{G_b \times G_f, \phi_b \times \phi_f}\isomorphic h^d_{f,(G_f,\phi_f)}\paren{BG_b, \phi_b}\label{GCH_iso_fermionic}
\end{equation}
that are natural in $(G_b,\phi_b)$.
\end{framednameddef}
}

{It is worth pointing out that Eq.\,(\ref{GCH_iso}) is equivalent the following statement. That is, for any given $(G_2,\phi_2)$, there exists a generalized cohomology theory $h_{b,(G_2,\phi_2)}$ such that there are isomorphisms
\begin{equation}
\SPT^d_b\paren{G_1 \times G_2, \phi_1 \times \phi_2} \isomorphic h^d_{b,(G_2,\phi_2)}(BG_1, \phi_1)
\end{equation}
that are natural in $(G_1, \phi_1)$ \cite{Gaiotto_Johnson-Freyd,Note61}. In other words, the existence of a single, ``universal" $h_b$ is equivalent to the existence of a family of generalized cohomology theories, one for each $(G_2,\phi_2)$; $h_b$ corresponds to the choice $G_2 = 0$. Stated this way, the bosonic Hypothesis is in complete analogy with the fermionic one, apart from the condition $\ZZZ_2^f \subset G_f$.}%
\phantom{
\setcounter{footnote}{60}
\footnote{{This follows from the canonical isomorphism
\begin{equation*}
\brackets{\tilde X \times \tilde Y, Z}_{\pi_1 X\times \pi_1 Y} \isomorphic \brackets{\tilde X, \Map^{\pi_1 Y}\paren{\tilde Y, Z} }_{\pi_1 X}
\end{equation*}
for any topological spaces $X$, $Y$, and $Z$, where $\tilde X$ and $\tilde Y$ are universal covers and $\Map^{\pi_1 Y}\paren{\tilde Y, Z}$ is the space of $\pi_1 Y$-equivariant maps from $\tilde Y$ to $Z$.}}
}

An important addendum to the Hypothesis, which extends the Hypothesis to not necessarily on-site symmetries, is this:

\begin{framednameddef}[Addendum to Hypothesis]
A symmetry should be represented antiunitarily if it reverses the orientation of spacetime, and unitarily otherwise.
\end{framednameddef}

\noindent Tensor-network \cite{Jiang_sgSPT} and spatial gauge field \cite{Thorngren_sgSPT}  arguments have been presented to justify this general rule. The rule is also consistent with numerous studies of SPT phases protected by specific spatial symmetries \cite{Wen_1d, Cirac, Wen_sgSPT_1d, Hsieh_sgSPT, Hsieh_sgSPT_2, You_sgSPT, SPt, Cho_sgSPT, Yoshida_sgSPT, Huang_dimensional_reduction}, especially Ref.\,\cite{Huang_dimensional_reduction}, which considered all 3D space-groups. The addendum implies that we should treat not only time reversal, but also spatial reflection, glide reflection, etc. as if they were antiunitary symmetries.

While this is technically not part of the minimalist framework, Kitaev has proposed a physical interpretation of the $\Omega$-spectrum \cite{Kitaev_Stony_Brook_2011_SRE_1, Kitaev_Stony_Brook_2013_SRE, Kitaev_IPAM}. Namely, $F_d$ should be the space of $d$-dimensional short-range entangled states. For instance, in the bosonic case, $F_0 \homotopic \CCC P^\infty$, which corresponds to rays in Hilbert spaces. Moreover, he proposed that condition (\ref{FOmegaF}) should encode a correspondence between short-range entangled states in adjacent dimensions that is based on a Jackiw-Rebbi soliton-type construction \cite{Kitaev_Stony_Brook_2011_SRE_1, Kitaev_Stony_Brook_2013_SRE, Kitaev_IPAM, Xiong}. The classifying space, on the other hand, signals a gauge theory nature of the low energy effective theories of SPT phases \cite{DijkgraafWitten, Kapustin_Boson, Kapustin_Fermion, Xiong}. Indeed, homotopy classes of maps from any space $X$ (e.g.\,a spacetime manifold) to $BG$ classifies principal $G$-bundles over $X$, which are nothing but a different name for gauge fields. Classifying spaces have also been shown to be an object that naturally arises when one classifies invertible phases of matter \cite{Freed_SRE_iTQFT,  Freed_ReflectionPositivity}.

\section{Relationship between the $\ZZZ_4$ and $\ZZZ_2$ classifications\label{app:relationship_Z4_Z2}}

The goal of this appendix is to prove \q{glideforget}, where $\chi \in \braces{0,1,2,3}$ is the $\ZZZ_4$ invariant distinguishing the four phases in the strong $\ZZZ_4$ classification of 3D glide-symmetric insulators in Wigner-Dyson class AII \cite{Nonsymm_Shiozaki, AA_Z4}, whereas $\nu_0 \in \braces{0,1}$ is the strong 3D $\ZZZ_2$ index in the absence of the glide constraint \cite{fu2007b,moore2007,Rahul_3DTI,Inversion_Fu}. $\chi$ was first defined in terms of the Berry connection and curvature in Ref.\ \cite{Nonsymm_Shiozaki} and later reformulated through holonomy by one of us in Ref.\ \cite{AA_Z4}; this latter formulation will be briefly reviewed and utilized to prove \q{glideforget}.

Let us pick a coordinate system where the glide symmetry $\bar{M}_y:(x,y,z){\rightarrow} (x+1/2,-y,z)$ is composed of a half lattice-translation in $\vec{x}$, and a reflection that inverts $\vec{y}$. The topological invariant $\chi$ is encoded in the parallel transport of Bloch wavefunctions along the noncontractible quasimomentum ($\bk$) loops illustrated in \fig{fig:wilson}(a-b). This family of loops lie within the bent, 2D quasimomentum subregion $abcd$, which resembles the surface of a rectangular pipe; faces $a$ and $c$ are each half of a glide-invariant plane. For each loop, $k_z$ is varied during the transport and $\bk_{\parallel}{=}(k_x,k_y)$ is held fixed; we let $t {\in} [0,4]$ (with $4{\equiv} 0$) parametrize $\bk_{\parallel}(t)$ [see \fig{fig:wilson}(b)]. The holonomy of this transport is represented by an $\noc$-by-$\noc$ matrix, where $\noc$ is the number of occupied bands; this matrix is also known as the Wilson loop ($\W$) of the non-abelian Berry gauge field \cite{wilczek1984}. Its unimodular eigenvalues $\{\exp[i\lambda_j(t)]|j{=}1,2,\ldots,\noc; t{\in} [0,4]\}$ are Berry-Zak phase factors which encode the multi-band holonomy \cite{berry1984,zak1989}; we refer to $\lambda_j$ as the quasienergy of a ``Wilson band" indexed by $j$. \\

\begin{figure}[t]
\centering
\includegraphics[width=8.3 cm]{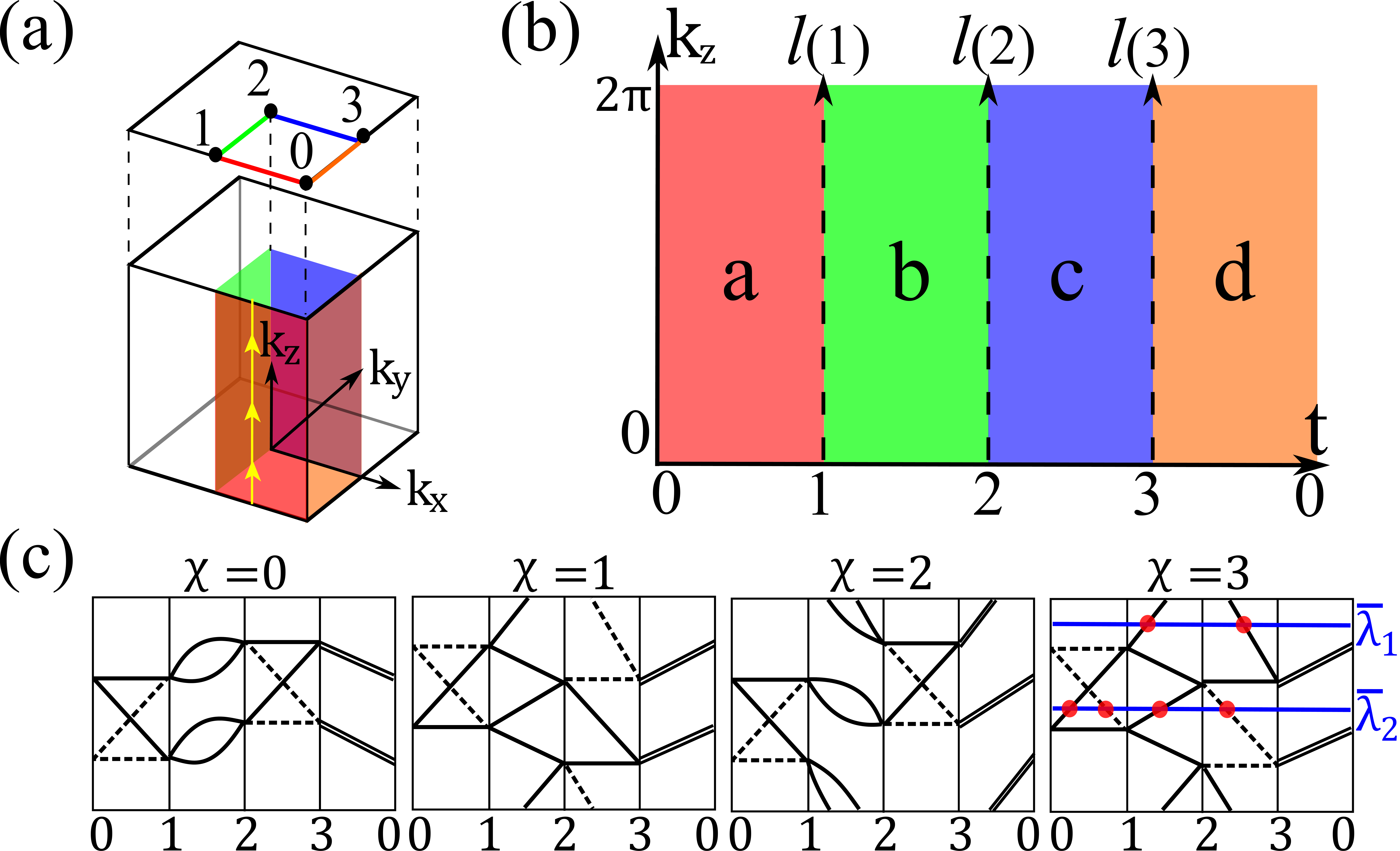}
\caption{(a) 3D Brillouin zone of glide-symmetric solids; red and blue faces are glide-invariant. (b) 2D bent subregion; the $\ZZZ_4$ invariant is defined over the red, blue and green faces. (c) Representative examples of glide-symmetric insulators in four classes distinguished by $\chi$. The vertical axis corresponds to the Berry-Zak phase $\lambda{\in}[0,2\pi]$.
}
\label{fig:wilson}
\end{figure}

While each of $\exp[i\lambda_j(t)]$ is invariant under $U(\noc)$ transformations in the space of occupied bands, $\lambda_j(t)$ is defined modulo $2\pi$, and a branch is chosen for each $j$ such that $\lambda_j$ is a piecewise continuous function of $t$ in each of the intervals $[0,1]$, $[1,2]$, and $[2,3]$ (or $01$, $12$, and $23$ for short). For the glide-invariant intervals $01$ and $23$ (which are the projections of faces $a$ and $c$), $\{\lambda_j^{\pm}\}_{\sma{j{=}1}}^{\sma{\noc/2}}$ are defined as the quasienergies of occupied bands in the even ($\Delta_+$) and odd ($\Delta_-$) representations of glide. In more detail, each wavevector in $a$ and $c$ is mapped to itself under glide, and therefore Bloch wavefunctions at each wavevector transform in two representations of glide [$\Delta_{\pm}(k_x)$] according to their glide eigenvalues $\pm i \exp[-ik_x/2]$, with $k_x$ in the first Brillouin zone.\\

The simplest way to identify $\chi$ through the quasienergy spectrum is to draw a constant-${\lambda}$ reference line [for an arbitrarily chosen quasienergy $\bar{\lambda}$, as illustrated by a blue line in \fig{fig:wilson}(c)] and consider its intersections with Wilson bands (indicated by red dots). For each intersection between $12$, we calculate the sign of the velocity $d\lambda/dt$, and sum this quantity over all intersections to obtain $\cals_{12}(\bar{\lambda})$; for $01$ and $23$, we consider only intersections with Wilson bands in the $\Delta_{\pm}$ representation, and we similary sum over sgn[$d\lambda/dt$] to obtain $\cals^{\pm}_{01}(\bar{\lambda})$ and $\cals^{\pm}_{23}(\bar{\lambda})$ respectively. The following \emph{weighted} sum of $\cals^{\pm}_{ij}$ and $\cals_{12}$, 
\e{\cals^{\pm}(\bar{\lambda}) = 2\cals^{\pm}_{01}(\bar{\lambda})+\cals_{12}(\bar{\lambda})+2\cals^{\pm}_{23}(\bar{\lambda}),\la{z4inv}}
satisfies $(\cals^{\pm}(\bar{\lambda}_1){-}\cals^{\pm}(\bar{\lambda_2}))/4{\in}\ZZZ$ for any $\bar{\lambda}_1$ and $\bar{\lambda}_2$; e.g., compare $\cals^{+}(\bar{\lambda}_1){=}2(0){+}1{+}2({-}1){=}{-}1$ with $\cals^{+}(\bar{\lambda}_2){=}2({+}1){+}1{+}2(0){=}3$ in \fig{fig:wilson}(c). Hence if we are only interested in mod-four equivalence classes, we may drop the argument of $\cals^{\pm}$ completely. It was further shown in Ref.\ \cite{AA_Z4} that $\cals^{+} \equiv {-}\cals^- \mod 4$, so there is only one independent $\ZZZ_4$ invariant: $\chi := \cals^+$.\\

Since both $\cals^{+}_{01}$ and $\cals^{+}_{23}\in \ZZZ$, it follows that 
\e{ \chi \equiv \cals_{12}(\bar{\lambda}) \mod 2.\la{chitos12}} 
To clarify, the parity of $\cals_{12}$ (hence also of $\chi$) is independent of the reference quasienergy $\bar{\lambda}$, and is even independent of the velocities of the Wilson bands that intersect an arbitrary quasienergy reference line. Parity($\cals_{12}$) merely counts the parity of the number of intersection points, and  
distinguishes between trivial ($\cals_{12} \equiv 0 \mod 2$) and topological ($\cals_{12} \equiv 1 \mod 2$) insulators in the 2D AII class \cite{kane2005A}; we refer here to the 2D classification because we have restricted the base space of the vector bundle to the $k_x=0$ plane (which contains the face $b$). In identifying the parity of $\cals_{12}$ as the topological invariant for class AII in 2D, we have utilized earlier works \cite{fu2006,yu2011,soluyanov2011} which identified a partner-switching in the Wannier centers of the Kane-Mele quantum-spin-Hall phase \cite{kane2005A}.\\

On the other hand, it is also known that the 3D-AII strong invariant ($\nu_0$) is expressible as the sum of 2D-AII strong invariants for any two inequivalent, parallel planes that are time reversal-invariant \cite{fu2007b,moore2007,Rahul_3DTI,Inversion_Fu}.
Particularizing this claim to the $k_x=0$ and ${-}\pi$ planes, we get
  \e{ \nu_0  \equiv \cals_{12}(\bar{\lambda})+\cals_{30}(\bar{\lambda}) \mod 2, \la{nu0toS12}}
where $\cals_{30}(\bar{\lambda})\in \{0,1\}$ is defined as the parity of Wilson bands that intersect the $\lambda = \bar \lambda$ reference line along $t\in[3,4]$. Finally, we utilize the fact that $\cals_{30}$ is always even for glide-symmetric insulators. Thus $\cals_{30}(\bar \lambda)$ drops out of Eq.\,(\ref{nu0toS12}) and the desired Eq.\,(\ref{glideforget}) follows. The evenness of $\cals_{30}$ follows from the two-fold degeneracy of Wilson bands for  $t\in[3,4]$. Indeed, any point on this interval is invariant under $T\bar{M}_y$, which squares to a Bravais-lattice translation in $\vec{x}$. Consequently, the representation of $(T\bar{M}_y)^2$ with  Bloch functions is $e^{-ik_x}=-1$, leading to a Kramers degeneracy for all  $t\in[3,4]$. This quick argument for degeneracy has been developed more precisely in Ref.\,\cite{Cohomological}. 

\section{Proof of Proposition \ref{prp:SES_classifications}\label{app:proof}}

{
Let $\phi: \ZZZ \times G \fromto \braces{\pm 1}$ be the homomorphism tracking unitarity versus antiunitarity of symmetry representations, and let us use the same symbol for its restriction to subgroups. By the Addendum to Hypothesis we know that $\phi$ sends the generator of $\ZZZ$ to $-1$. For bosonic SPT phases, by the Hypothesis we have
\begin{eqnarray}
\SPT^d_b\paren{G,\phi} &\isomorphic& h^d_b\paren{BG,\phi}, \\
\SPT^d_b\paren{\ZZZ \times G,\phi} &\isomorphic& h^d_b\paren{B(\ZZZ\times G),\phi}.
\end{eqnarray}
For fermionic SPT phases, choosing $G_f = G$ in the factorization of both $G$ and $\ZZZ \times G$ [see Eq.\,(\ref{G=G_b_times_G_f})], we have
\begin{eqnarray}
\SPT^d_f\paren{G,\phi} &\isomorphic& h^d_{f,(G_f,\phi)}\paren{B0, \phi}, \\
\SPT^d_f\paren{\ZZZ \times G,\phi} &\isomorphic& h^d_{f,(G_f,\phi)}\paren{B\paren{\ZZZ \times 0},\phi},
\end{eqnarray}
where $h_{f,(G_f,\phi)}$ is the generalized cohomology theory associated with $(G_f,\phi)$ (see App.\,\ref{app:TGCH}), and $0$ denotes the trivial group. Since both $h_b$ and $h_{f,(G_f,\phi)}$ are generalized cohomology theories, to prove Proposition \ref{prp:SES_classifications} it will suffice, in either case, to establish the following property of an arbitrary generalized cohomology theory:

\begin{widetext}
\begin{prp}\label{prp:SES_h}
Let $h$ be any generalized cohomology theory, $n$ be any integer, $G$ be any group, and $\phi: \ZZZ \times G \fromto \braces{\pm 1}$ be any homomorphism sending the generator of $\ZZZ$ to $-1$. There is a short exact sequence
\begin{eqnarray}
0 \fromto h^{n-1}\paren{BG,\phi}/2h^{n-1}\paren{BG,\phi} \xfromto{\alpha} h^n\paren{B\paren{\ZZZ\times G},\phi} \xfromto{\beta} \braces{[c]\in h^n\paren{BG,\phi} | 2[c] = 0} \fromto 0,
\end{eqnarray}
where $\beta$ is induced by the inclusion $G \fromto \ZZZ \times G$.
\end{prp}
\end{widetext}

Any group homomorphism $f:G_1 \fromto G_2$ gives rise to a short exact sequence that reads \cite{artin2011algebra}
\begin{equation}
0 \fromto \kernel f \fromto G_1 \fromto \image f \fromto 0, \label{temp007}
\end{equation}
where the first map $\kernel f \fromto G_1$ is the inclusion and the second map $G_1 \fromto \image f$ is the same as $f$ but with possibly restricted codomain. According to Eqs.\,(\ref{temp020})(\ref{temp021}), the inclusion $G \fromto \ZZZ \times G$ induces a homomorphism
\begin{equation}
\beta': h^n\paren{B\paren{\ZZZ \times G},\phi} \fromto h^n\paren{BG, \phi}
\end{equation}
on generalized cohomology groups. Consequently, there is a short exact sequence
\begin{equation}
0 \fromto \kernel \beta' \fromto h^n\paren{B\paren{\ZZZ \times G},\phi} \fromto \image \beta' \fromto 0.
\end{equation}
To prove Proposition \ref{prp:SES_h}, the idea is to show that
\begin{eqnarray}
\image \beta' &=& \{ [c] \in h^n\paren{BG,\phi} \big| 2[c] = 0 \}, \label{proof_image} \\
\kernel \beta' &\isomorphic&  h^{n-1}\paren{BG} / 2 h^{n-1}\paren{BG}. \label{proof_kernel}
\end{eqnarray}

Let us represent $h$ by an $\Omega$-spectrum $\paren{F_n}_{n\in \ZZZ}$.
According to Eq.\,(\ref{widetilde_X_Omega_F_n+1}), we have
\begin{equation}
h^n\paren{BG, \phi} \coloneq \brackets{\widetilde{BG}, \Omega F_{n+1}}_{G}. \label{temp001}
\end{equation}
By definition, the loop space $\Omega F_{n+1}$ is the space of pointed maps from the circle $\SS^1$ to $F_{n+1}$ \cite{Hatcher}, i.e.
\begin{equation}
\Omega F_{n+1} \coloneq \Map_\star\paren{\SS^1, F_{n+1}}.
\end{equation}
We can regard $\SS^1 \coloneq I/\partial I$ as the unit interval $I \coloneq [0,1]$ with its endpoints identified ($\partial I \coloneq \braces{0,1}$ denotes the boundary of $I$), and move it to the other side of $[-,-]_G$ in Eq.\,(\ref{temp001}). That is,
\begin{equation}
h^n\paren{BG, \phi} \isomorphic \angles{ \paren{I \times \widetilde{BG}} / \paren{\partial I \times \widetilde{BG}}, F_{n+1} }_G. \label{temp015}
\end{equation}
Here, $\angles{-,-}_G$ is the same as $\brackets{-,-}_G$ but requires basepoint \cite{munkres2000topology} to be preserved: it is the set of homotopy classes of \emph{pointed}, $G$-equivariant maps. In turn, let us tuck $\widetilde{BG}$ away in the codomain of $\angles{-,-}_G$:
\begin{equation}
h^n\paren{BG, \phi} \isomorphic \angles{ I / \partial I, \Map_\star^G\paren{ \widetilde{BG}, F_{n+1} } }. \label{temp004}
\end{equation}
where $\Map_\star^G\paren{-,-}$ is the same as $\Map_\star\paren{-,-}$ but requires $G$-action to be preserved; it is the space of pointed, \emph{$G$-equivariant} \cite{DavisKirk} maps. Note that the $G$-equivariance requirement disappeared from $\angles{-,-}_G$ as a side effect \cite{Note71}. Introducing the shorthand
\begin{equation}
\paren{Y, y_0} \coloneq \Map_\star^G\paren{ \widetilde{BG}, F_{n+1} },
\end{equation}
where $y_0$ denotes the basepoint of $Y$, we now have
\begin{equation}
h^n\paren{BG, \phi} \isomorphic \angles{ I / \partial I, Y }. \label{hdGY}
\end{equation}
A similar derivation shows that
\begin{equation}
h^n\paren{B\paren{\ZZZ \times G},\phi} \isomorphic \angles{ \paren{I \times \RRR} / \paren{\partial I \times \RRR}, Y }_{\ZZZ}, \label{hdZGY}
\end{equation}
where $\RRR$ arose as the universal cover of the classifying space $\SS^1$ of $\ZZZ$ \cite{AdemMilgram}. Eq.\,(\ref{hdZGY}) is the set (actually, abelian group) of homotopy classes of $\ZZZ$-equivariant, pointed maps from $\paren{I \times \RRR} / \paren{\partial I \times \RRR}$ to $Y$. If we introduce coordinates $s$ and $r$ for $I$ and $\RRR$, respectively, then the generator of $\ZZZ$ will send $(s,r)$ to $(1-s, r+1)$. The action of $\ZZZ$ on $Y$ is trivial \cite{Note81}.%
\phantom{
\setcounter{footnote}{70}
\footnote{{With the formulation in App.\,\ref{app:TGCH}, $G$ would act nontrivially on $I$, and the $G$-equivariance requirement would not disappear from $\angles{-,-}_G$ as we pass from Eq.\,(\ref{temp015}) to Eq.\,(\ref{temp004}). Fortunately, there is an equivalent formulation where $G$ acts nontrivially on $F_{n+1}$ instead of $I$, for which the $G$-equivariance requirement does disappear.}}
\setcounter{footnote}{80}
\footnote{{Technically, in view of footnote \cite{Note71}, we ought to let the generator of $\ZZZ$ send $(s,r)$ to $(s,r+1)$ and act nontrivially on $Y$, but this is essentially equivalent to the alternative action considered here and would only incur minor changes to the proof.}}
}

With Eqs.\,(\ref{hdGY}) and (\ref{hdZGY}), we can now forget about $G$ and work directly with $Y$. Note that Eqs.\,(\ref{hdGY}) and (\ref{hdZGY}) are the codomain and domain of the homomorphism $\beta'$. Explicitly, an element of Eq.\,(\ref{hdZGY}) is represented by a pointed map
\begin{eqnarray}
b &:& \paren{I \times \RRR} / \paren{\partial I \times \RRR} \fromto Y, \nonumber\\
&& (s,r) \mapsto b(s,r) \label{temp002}
\end{eqnarray}
such that
\begin{eqnarray}
b(s,r) = b(1-s,r+1), ~\forall s, r, \label{temp003}
\end{eqnarray}
whereas an element of Eq.\,(\ref{hdGY}) is represented by a pointed map
\begin{eqnarray}
c &:& I/\partial I \fromto Y, \nonumber\\
&& s \mapsto c(s).
\end{eqnarray}
See Figs.\,\ref{fig:proof_maps}(b)(c) for an illustration.
Zero elements of the abelian groups (\ref{hdGY}) and (\ref{hdZGY}) are represented by constant maps, which map all points to $y_0$. The homomorphism $\beta'$ itself can be given by restricting to $r = 0$, that is, by setting
\begin{equation}
c(s) = b(s,0),
\end{equation}
as depicted in Fig.\,\ref{fig:proof_image_1}.

\begin{figure}[t]
\centering
\includegraphics[width=3in]{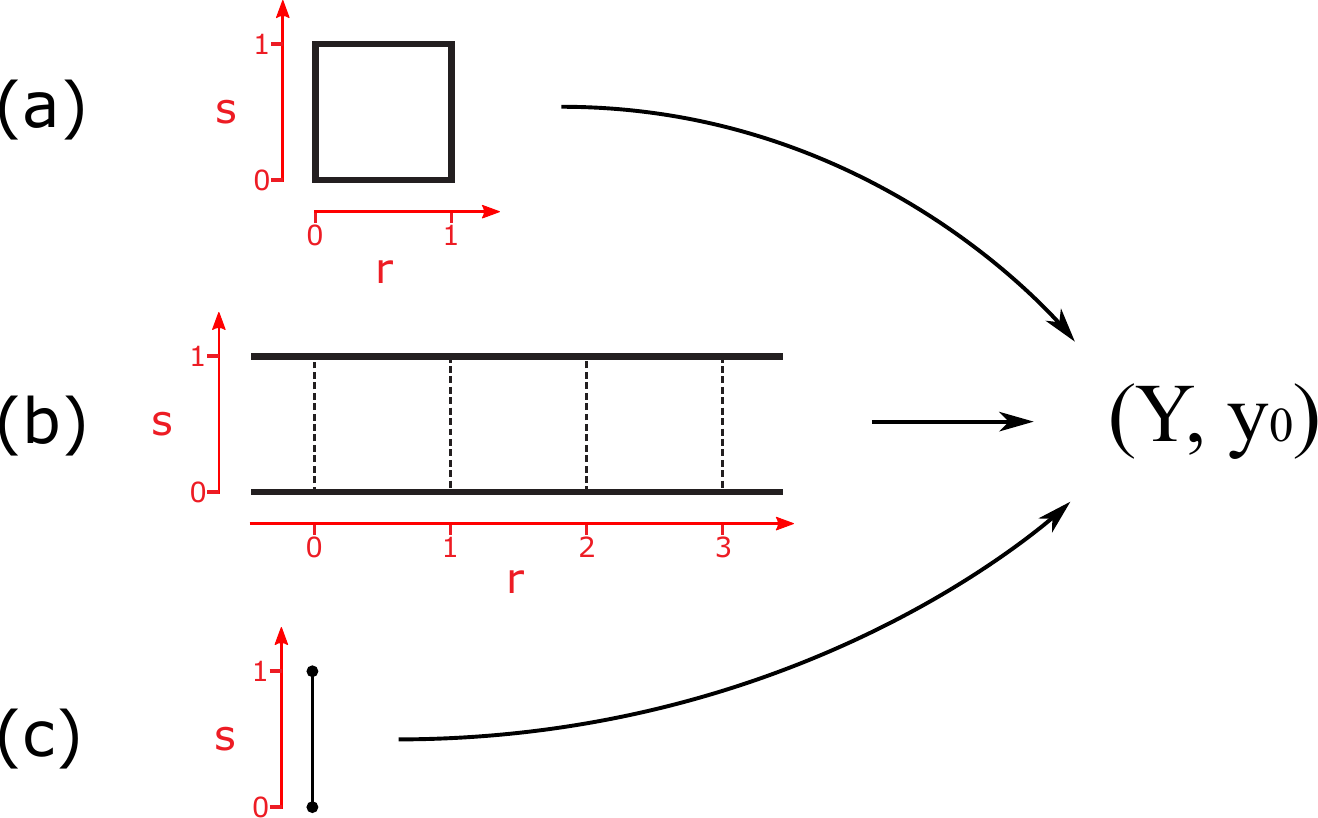}
\caption{{Elements of $h^{n-1}\paren{BG, \phi}$, $h^n\paren{B\paren{\ZZZ \times G},\phi}$, and $h^n\paren{BG, \phi}$ can be represented by maps from (a) $I \times I$, (b) $I \times \RRR$, and (c) $I$ to $Y$, respectively.}}
\label{fig:proof_maps}
\end{figure}

We are now ready to prove Eq.\,(\ref{proof_image}). Take any $b$ as in Eqs.\,(\ref{temp002})(\ref{temp003}) and let $c: I/\partial I \fromto Y$ be its restriction to $r=0$. Dialing $r$ from $0$ to $1$ then yields a homotopy $c \sim \bar c$ that preserves basepoint, where $\bar c(s) \coloneq c(1-s)$. This means that $[c] = -[c]$, or equivalently $2[c] = 0$, as elements of $h^n\paren{BG, \phi}$. Conversely, let $c: I/\partial I \fromto Y$ be any pointed map that is homotopic to $\bar c$ while preserving basepoint. Write $c_r: I/\partial I \fromto Y$ for the homotopy, with $r \in [0,1]$, $c_0 = c$, and $c_1 = \bar c$. Then the $b$ defined by
\begin{equation}
b(s,r) = \begin{cases}
c_{r-2n}(s), & r\in [2n, 2n+1],~n\in \ZZZ,\\
c_{r-2n+1}(1-s) , & r\in [2n-1, 2n],~n\in \ZZZ
\end{cases}
\end{equation}
will be a pointed map (\ref{temp002}) satisfying Eq.\,(\ref{temp003}) whose restriction to $r=0$ is $c$.

As for Eq.\,(\ref{proof_kernel}), we note that we have
\begin{equation}
h^{n-1}\paren{BG, \phi} \isomorphic \angles{ I/\partial I, \Map_\star^G\paren{\widetilde{BG}, F_n} },
\end{equation}
by shifting the degree in Eq.\,(\ref{temp004}). Thanks to Eq.\,(\ref{FOmegaF}), we can further replace $F_n$ by $\Omega F_{n+1}$ and trade $\Omega$ for an extra factor of $I$ in the domain of $\angles{-,-}$. This gives
\begin{equation}
h^{n-1}\paren{BG, \phi} \isomorphic \angles{ (I \times I)/ \partial (I \times I), Y }.\label{hd-1GY}
\end{equation}
See Fig.\,\ref{fig:proof_maps}(a). We define a map,
\begin{equation}
\alpha': h^{n-1}\paren{BG, \phi} \fromto  h^n\paren{B\paren{\ZZZ \times G},\phi}, \label{alpha'_h}
\end{equation}
by identifying the two factors of $I$ in $(I \times I)/ \partial (I \times I)$ with $I$ and $\RRR$ in $\paren{I \times \RRR} / \paren{\partial I \times \RRR}$, where we extend the second $I$ to $\RRR$ via the $\ZZZ$-action.
More explicitly, if a pointed map
\begin{equation}
a: (I \times I)/ \partial (I \times I)\fromto Y \label{temp005}
\end{equation}
represents an element $[a]$ of $h^{n-1}\paren{BG, \phi}$, then the $b$ defined by
\begin{equation}
b(s,r) = \begin{cases}
a(s,r-2n), & r\in [2n, 2n+1],~n\in \ZZZ,\\
a(1-s,r-2n+1) , & r\in [2n-1, 2n],~n\in \ZZZ
\end{cases}\label{temp012}
\end{equation}
will represent $\alpha'([a])$. See Fig.\,\ref{fig:proof_a_b} for an illustration. This $\alpha'$ is a homomorphism because addition can be defined using $Y$.

\begin{figure}[t]
\centering
\includegraphics[width=2in]{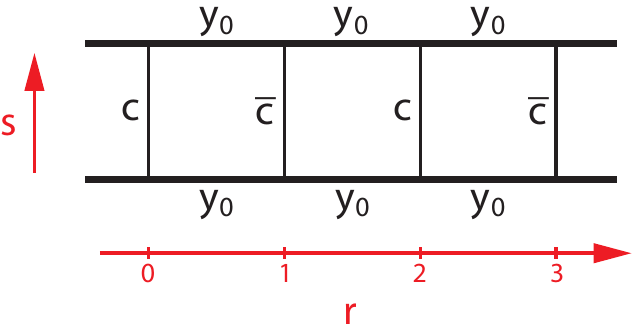}
\caption{{A representative $b: \paren{I \times \RRR} / \paren{\partial I \times \RRR} \fromto Y$ of an element of Eq.\,(\ref{hdGY}) and its restriction $c: I/\partial I \fromto Y$ to $r=0$.}}
\label{fig:proof_image_1}
\end{figure}

\begin{figure}[t]
\centering
\includegraphics[width=2.4in]{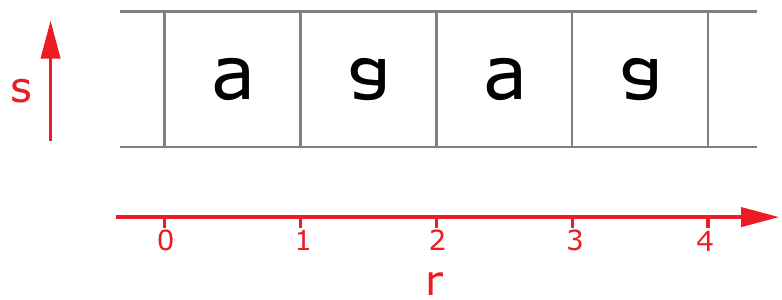}
\caption{{A representative $a: (I \times I)/ \partial (I \times I)\fromto Y$ of an element of  Eq.\,(\ref{hd-1GY}) and the map $b: \paren{I \times \RRR} / \paren{\partial I \times \RRR} \fromto Y$ defined by Eq.\,(\ref{temp012}).}}
\label{fig:proof_a_b}
\end{figure}

\begin{figure}[t]
\centering
\includegraphics[width=2in]{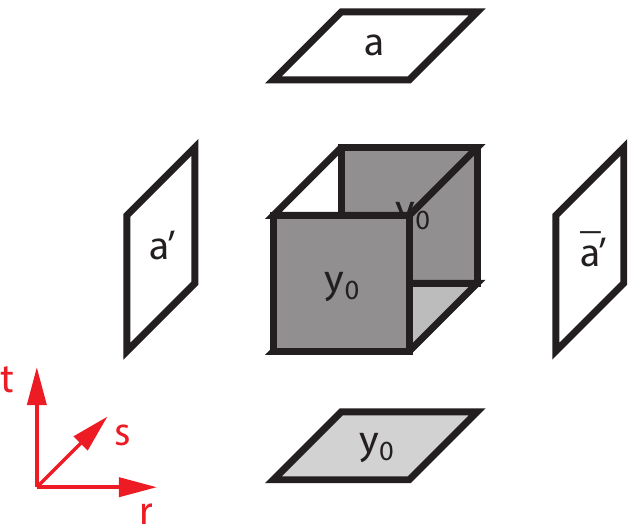}
\caption{A homotopy $a_\bullet$ as in Eqs.\,(\ref{aIIY})-(\ref{aIIY_p3}). Its restrictions to $t=1$, $r=0$, $r=1$, and $t=0$ are $a$, $a'$, $\overline{a'}$, and the constant map at $y_0$, respectively.}
\label{fig:proof_kernel_1}
\end{figure}

To prove Eq.\,(\ref{proof_kernel}), we first note that $\image \alpha'$ must be contained in $\kernel \beta'$. This is because a pointed map (\ref{temp005}) by definition maps the entire $\partial (I \times I)$ to the basepoint $y_0$ of $Y$. Conversely, $\kernel \beta'$ must be contained in $\image \alpha'$. Indeed, if the restriction $c$ of a given $b$ to $r=0$ is homotopic (while preserving basepoint) to the constant map at $y_0$, then $b$ itself must be homotopic (while preserving basepoint and $\ZZZ$-action) to a map whose restriction to $r=0$ is the constant map at $y_0$; that is, the restriction to $r=0$ satisfies the homotopy extension property \cite{Hatcher,DavisKirk}. Therefore, we have
\begin{equation}
\image \alpha' = \kernel \beta'. \label{temp008}
\end{equation}
This boils down the proof of Eq.\,(\ref{proof_kernel}) to the computation of $\image \alpha'$.  

We will determine $\image \alpha'$ by combining Eq.\,(\ref{temp007}) with an explicit computation of $\kernel \alpha'$. Suppose a pointed map $a: (I \times I)/ \partial(I \times I) \fromto Y$ represents an element of $\kernel \alpha'$. This implies that there is a homotopy,
\begin{equation}
a_t: I \times I \fromto Y, \label{aIIY}
\end{equation}
with $t \in [0,1]$, such that
\begin{eqnarray}
&&a_0(s,r) = y_0, ~\forall s, r, \label{aIIY_p1}\\
&&a_t(0,r) = a_t(1,r) = y_0, ~\forall t, r, \label{aIIY_p2} \\
&&a_1(s, r) = a(s,r), ~\forall s, r. \label{aIIY_p3}
\end{eqnarray}
We can visualize this homotopy $a_\bullet$ by treating $t$, $s$, and $r$ as the three coordinates of a hypothetical cube, shown in Fig.\,\ref{fig:proof_kernel_1}. Now, the key observation is that 
the restriction of $a_\bullet$ to $r=0$ can also be viewed as a pointed map $(I \times I)/ \partial(I \times I) \fromto Y$, which we shall denote by $a'$. The restriction of $a_\bullet$ to $r=1$ will then be $\overline{a'}$, where $\overline{a'}$ is the same as $a$ but has the opposite orientation. As a result, $a$ is homotopic (while preserving basepoint) to $a'+a'$, and $[a] = [a'] + [a']$ as elements of $h^{n-1}\paren{BG, \phi}$. That is, $[a] \in 2 h^{n-1}\paren{BG,\phi}$ [see Eq.\,(\ref{2A}) for the meaning of 2]. This shows that $\kernel \alpha' \subset 2 h^{n-1}\paren{BG,\phi}$. Conversely, if $[a] = [a'] + [a']$ for some $a'$, then one can construct an $a_\bullet$ that satisfies Eqs.\,(\ref{aIIY_p1})(\ref{aIIY_p2})(\ref{aIIY_p3}) by putting $a'$ and $\overline{a'}$ on $r=0$ and $r=1$, respectively. This shows that $2 h^{n-1}\paren{BG,\phi} \subset \kernel \alpha'$. Therefore, we have
\begin{equation}
\kernel \alpha' = 2 h^{n-1}\paren{BG,\phi}. \label{temp009}
\end{equation}
To complete the proof, we note that
\begin{eqnarray}
\kernel \beta' &=& \image \alpha' \nonumber\\
&\isomorphic& h^{n-1}\paren{BG,\phi} / \kernel \alpha' \nonumber\\
&=& h^{n-1}\paren{BG,\phi} / 2 h^{n-1}\paren{BG,\phi},
\end{eqnarray}
where we have used Eqs.\,(\ref{temp008}), (\ref{temp007}), and (\ref{temp009}) in the first, second, and third lines, respectively.}

\section{Proof of Corollaries \ref{cor:quad-chotomy} and \ref{cor:direct_sum_decomposition} \label{app:proof_corollaries}}

\begin{proof}[Proof of Corollary \ref{cor:quad-chotomy}]
It is clear that the four types are mutually exclusive. Take any $[b]\in \SPT^d\paren{\ZZZ \times G}$. If $[b]=0$, then it is of type (i). Now assume $[b]\neq 0$. There are two cases: either $\beta ([b]) = 0$ or $\beta([b]) \neq 0$. In the first case, $[b]\in \kernel \beta = \image \alpha$, so there is an $[a]\in A$ such that $\alpha([a]) = [b]$. This means $2[b] = \alpha(2[a]) = \alpha(0) = 0$. Thus $[b]$ is of type (ii). In the second case, we have $\beta(2[b]) = 2 \beta([b]) = 0$. By the same argument, $2(2[b]) = 4[b] = 0$. It follows that either $2[b] = 0$ and $[b]$ is of type (iii), or $2[b] \neq 0$, in which case we know by the same argument that $2[b]$ is of type (ii) and so $[b]$ is of type (iv).
\end{proof}

\begin{proof}[Proof of Corollary \ref{cor:direct_sum_decomposition}]
By Corollary \ref{cor:quad-chotomy}, any nontrivial element of $\SPT^d\paren{\ZZZ \times G}$ must have order 2 or 4.
The proof is trivial in case $\SPT^d\paren{\ZZZ \times G}$ is finitely generated, thanks to the fundamental theorem of finitely generated abelian groups.
In general, we note that the first remark in the proof implies that $\SPT^d\paren{\ZZZ \times G}$ is an abelian $p$-group for $p=2$ with bounded exponent [the smallest positive integer $k$ such that $k [b] = 0$ for all $[b]\in \SPT^d\paren{\ZZZ \times G}$]. The desired result then follows from Pr\"ufer's first theorem, which states that an abelian $p$-group with bounded exponent must be a (possibly infinite, or even uncountable) direct sum of cyclic subgroups \cite{Kargapolov}.
\end{proof}

\section{Proof of Proposition \ref{prp:glide_to_translation}\label{app:proof_translation}}

Proposition \ref{prp:glide_to_translation} amounts to the statement that {the kernel of Eq.\,(\ref{beta}) is contained in the kernel of Eq.\,(\ref{gamma}). Similarly to App.\,\ref{app:proof}, for bosonic SPT phases, by the Hypothesis we have
\begin{eqnarray}
\SPT^d_b\paren{G,\phi} &\isomorphic& h^d_b\paren{BG,\phi}, \\
\SPT^d_b\paren{\ZZZ \times G,\phi} &\isomorphic& h^d_b\paren{B(\ZZZ\times G),\phi}, \\
\SPT^d_b\paren{2\ZZZ \times G,\phi} &\isomorphic& h^d_b\paren{B(2\ZZZ\times G),\phi},
\end{eqnarray}
For fermionic SPT phases, choosing $G_f = G$ in the factorization of $G$, $\ZZZ \times G$, and $2\ZZZ \times G$ [see Eq.\,(\ref{G=G_b_times_G_f})], we have
\begin{eqnarray}
\SPT^d_f\paren{G,\phi} &\isomorphic& h^d_{f,(G_f,\phi)}\paren{B0, \phi}, \\
\SPT^d_f\paren{\ZZZ \times G,\phi} &\isomorphic& h^d_{f,(G_f,\phi)}\paren{B\paren{\ZZZ \times 0},\phi}, \\
\SPT^d_f\paren{2\ZZZ \times G,\phi} &\isomorphic& h^d_{f,(G_f,\phi)}\paren{B\paren{2\ZZZ \times 0},\phi},
\end{eqnarray}
where $h_{f,(G_f,\phi)}$ is the generalized cohomology theory associated with $(G_f,\phi)$ (see App.\,\ref{app:TGCH}), and $0$ denotes the trivial group. Since both $h_b$ and $h_{f,(G_f,\phi)}$ are generalized cohomology theories, to prove Proposition \ref{prp:glide_to_translation} it will suffice, in either case, to establish the following property of an arbitrary generalized cohomology theory:
\begin{prp}\label{prp:glide_to_translation_h}
Let $h$ be any generalized cohomology theory, $n$ be any integer, $G$ be any group, and $\phi: \ZZZ \times G \fromto \braces{\pm 1}$ be any homomorphism sending the generator of $\ZZZ$ to $-1$. Let
\begin{eqnarray}
&&\beta' : h^n\paren{B\paren{\ZZZ\times G},\phi} \fromto h^n\paren{BG,\phi}, \\
&&\gamma: h^n\paren{B\paren{\ZZZ\times G},\phi} \fromto h^n\paren{B\paren{2\ZZZ \times G},\phi}
\end{eqnarray}
be induced by the inclusions $G \fromto \ZZZ \times G$ and $2\ZZZ \times G \fromto \ZZZ \times G$. Then
\begin{equation}
\kernel \beta' \subset \kernel \gamma. \label{temp011}
\end{equation}
\end{prp}

\begin{figure}[t]
\centering
\includegraphics[width=1.8in]{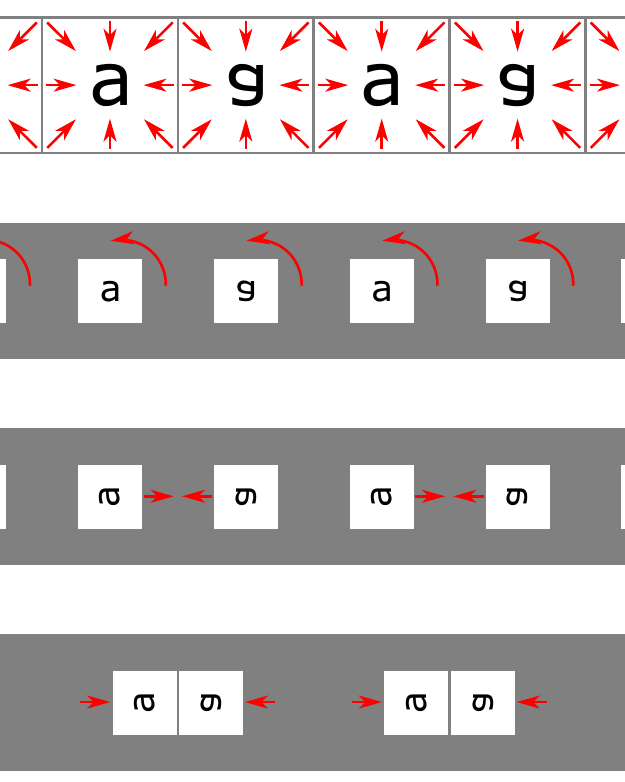}
\caption{A homotopy of map (\ref{temp012}) to the constant map at $y_0$ while preserving basepoint respecting condition (\ref{relaxed_condition_b}).}
\label{fig:proof_kernel_gamma}
\end{figure}

Let $\alpha'$ be defined as in Eq.\,(\ref{alpha'_h}), and let
\begin{equation}
\delta: h^n\paren{B\paren{2\ZZZ \times G},\phi} \fromto h^n\paren{BG,\phi}
\end{equation}
be induced by the inclusion $G \fromto 2\ZZZ \times G$. By the functoriality of a generalized cohomology theory [see remarks following Eq.\,(\ref{temp010})], we must have $\beta' = \delta \circ \gamma$. This implies that
\begin{equation}
\kernel \gamma \subset \kernel \beta'. \label{temp017}
\end{equation}
Below, we will show that
\begin{equation}
\image \alpha' \subset \kernel \gamma. \label{temp018}
\end{equation}
Combined with the equality $\image \alpha' = \kernel \beta'$ established in App.\,\ref{app:proof}, Eqs.\,(\ref{temp017})(\ref{temp018}) would then imply that $\image \alpha'$, $\kernel \gamma$, and $\kernel \beta'$ are actually all the same, from which Proposition \ref{prp:glide_to_translation_h} would follow.

To show $\image \alpha' \subset \kernel \gamma$, or equivalently $\gamma \circ \alpha' = 0$, we note that we can express
\begin{eqnarray}
h^n\paren{B\paren{2\ZZZ \times G},\phi} \isomorphic \angles{ \paren{I \times \RRR} / \paren{\partial I \times \RRR}, Y }_{2\ZZZ}, \label{hd2ZGY}
\end{eqnarray}
in much the same vein as Eqs.\,(\ref{hdGY}) and (\ref{hdZGY}). Here, $2\ZZZ$ is the subgroup of $\ZZZ$ of even integers, which means its generator is twice the generator of $\ZZZ$. Like for Eq.\,(\ref{hdZGY}), an element of Eq.\,(\ref{hd2ZGY}) is represented by a pointed map
\begin{eqnarray}
b &:& \paren{I \times \RRR} / \paren{\partial I \times \RRR} \fromto Y, \nonumber\\
&& (s,r) \mapsto b(s,r).
\end{eqnarray}
This map, however, only needs to satisfy the weaker condition,
\begin{equation}
b(s,r) = b(s, r+2), ~\forall s, r, \label{relaxed_condition_b}
\end{equation}
associated with twice the generator of $\ZZZ$. The homomorphism $\gamma$ corresponds to the relaxation of constraint (\ref{temp003}) to constraint (\ref{relaxed_condition_b}).

Now, take any $[a] \in h^{n-1}\paren{BG, \phi}$ and let $[b]$ be the element $\alpha'([a])$ of $h^n\paren{B\paren{\ZZZ \times G}, \phi}$. Recall from App.\,\ref{app:proof} that we can represent $[a]$ by a pointed map $a$ of the form (\ref{temp005}), and $[b]$ by corresponding the map $b$ defined in Eq.\,(\ref{temp012}); see Fig.\,\ref{fig:proof_a_b}. Depending on whether constraint (\ref{temp003}) or (\ref{relaxed_condition_b}) is imposed, this $b$ can be viewed as representing either the element $[b]$ of $h^n\paren{B\paren{\ZZZ \times G}, \phi}$ or the element $\gamma([b])$ of $h^n\paren{B\paren{2\ZZZ \times G}, \phi}$. In general, it may not be possible to homotope (i.e.\,deform) $b$ to the constant map at $y_0$ while preserving basepoint and respecting constraint (\ref{temp003}); this reflects the fact that $[b]$ may be nontrivial. However, if constraint (\ref{temp003}) is relaxed to (\ref{relaxed_condition_b}), then a homotopy always exists. Indeed, as illustrated in Fig.\,\ref{fig:proof_kernel_gamma}, starting with the picture in Fig.\,\ref{fig:proof_a_b}, we can first shrink the squares of $a$ (and its upside-down partners), rotate the shrunk squares, moving pairs of the rotated squares towards each other, and finally annihilating them. This shows that $\gamma([b])$ is the trivial element of $h^n\paren{B\paren{2\ZZZ \times G}, \phi}$, and so $\gamma \circ \alpha' = 0$.
}




\bibliographystyle{apsrev4-1}


%


\end{document}